%% file: main.tex
\documentclass{article}

\usepackage{PRIMEarxiv}

\usepackage[utf8]{inputenc} %
\usepackage[T1]{fontenc}    %
\usepackage{booktabs}       %
\usepackage{amsfonts}       %
\usepackage{nicefrac}       %
\usepackage{microtype}      %
\usepackage{fancyhdr}       %
\usepackage{graphicx}       %
\graphicspath{{Fig/}}     %

\usepackage{setspace}
\usepackage{enumitem}
\usepackage{amsmath}
\usepackage{bbm}
\usepackage{mathtools}
\usepackage{harpoon}
\usepackage{xspace}
\newcommand*{\scale}[2][4]{\scalebox{#1}{$#2$}}%
\usepackage{xcolor}
\usepackage{multicol}
\usepackage{algorithm}%
\usepackage{algorithmicx}%
\usepackage{algpseudocode}%
\usepackage{listings}%
\usepackage[toc,page]{appendix}

\definecolor{oceanboatblue}{rgb}{0.0, 0.47, 0.75}
\definecolor{officegreen}{rgb}{0.0, 0.5, 0.0}
\usepackage[colorlinks=true,allcolors=black,urlcolor=officegreen,citecolor=oceanboatblue]{hyperref}       %
\usepackage{url}            %

\usepackage{amsthm}
\newtheorem{theorem}{Theorem}
\newtheorem{proposition}{Proposition}
\newtheorem{lemma}{Lemma}

\theoremstyle{definition}
\newtheorem{definition}{Definition}

\usepackage[authoryear,round]{natbib}
\usepackage[capitalize,noabbrev]{cleveref}

\newlist{properties}{enumerate}{10}

\crefname{propertiesi}{Property}{Properties}
\Crefname{propertiesi}{Property}{Properties}

\crefname{example}{Example}{Examples}
\Crefname{example}{Example}{Examples}

\usepackage[verbose,hyperpageref]{backref}

\input{defs}

\definecolor{CommentColor}{rgb}{0,.50,.50}
\newcounter{margincounter}

\title{Non-parametric Hypothesis Tests\\ for Distributional Group Symmetry
}

\author{
  Kenny Chiu \\
  Department of Statistics \\
  The University of British Columbia \\
  \texttt{kenny.chiu@stat.ubc.ca} \\
   \And
  Benjamin Bloem-Reddy \\
  Department of Statistics \\
  The University of British Columbia \\
  \texttt{benbr@stat.ubc.ca} \\
}

\begin{document}

\maketitle

\begin{abstract}
    Symmetry plays a central role in the sciences, machine learning, and statistics. For situations in which data are known to obey a symmetry, a multitude of methods that exploit symmetry have been developed. Statistical tests for the presence or absence of general group symmetry, however, are largely non-existent. 
    This work formulates non-parametric hypothesis tests, based on a single independent and identically distributed sample, for distributional symmetry under a specified group. 
    We provide a general formulation of tests for symmetry that apply to two broad settings. The first setting tests for the invariance of a marginal or joint distribution under the action of a compact group.
    Here, an asymptotically unbiased test only requires a computable metric on the space of probability distributions and the ability to sample uniformly random group elements.
    Building on this, we propose an easy-to-implement conditional Monte Carlo test and prove that it achieves exact $p$-values with finitely many observations and Monte Carlo samples.
    The second setting tests for the invariance or equivariance of a conditional distribution under the action of a locally compact group.
    We show that the test for conditional invariance or equivariance can be formulated as particular tests of conditional independence.
    We implement these tests from both settings using kernel methods and study them empirically on synthetic data. Finally, we apply them to testing for symmetry in geomagnetic satellite data and in two problems from high-energy particle physics. 
\end{abstract}

\section{Introduction} \label{sec:intro}

Symmetry has played an important role in statistical problems, from the classical literature on equivariant estimation \citep[][Ch.~3]{lehmann_theory_1998} and invariant testing \citep[][Ch.~6]{Lehmann:Romano:2005}, to modern work on the use of transformation groups in statistical estimation \citep{Chen:2020,Huang:2022aa} and machine learning problems \citep{Cohen:2019}.
One of the key ideas that emerges from this line of work is that by using models that account for symmetries present in data, one obtains statistical benefits through various forms of optimality \citep{lehmann_theory_1998,Lehmann:Romano:2005,eaton:sudderth:1999:consitency}, improved sample efficiency \citep{Chen:2020,Huang:2022aa}, and better out-of-sample generalization \citep{Elesedy:2021:equivariant,Elesedy:2021,Lyle:2020}.
Such approaches also have a certain appeal that bridges different philosophical positions: decision problems with symmetry are among the only known problems in which frequentist and Bayesian inference coincide, and also present working realizations of other approaches to inference, variously known as fiducial \citep{Fraser:1961:fiducial:invariance,hora:buehler:fiducial:estimation,hora:buehler:prediction}, pivotal \citep{eaton:sudderth:1999:consitency}, or structural \citep{fraser:1966:structural,fraser:1968:structure:book}.
A pervasive characteristic shared by all of that work is that a specific symmetry group is known or assumed, and the problem is carefully constructed with respect to that group.
That work, however, does not address the problem of identifying symmetry from data. 
Moreover, a symmetry assumption can be difficult to check, and if the assumption is wrong, then the performance of a symmetric model can be much worse than a non-symmetric one. 

Separately, symmetry plays a central role in modern science, particularly in the physical sciences where entire theories are constructed around the symmetries that must be obeyed by equations describing the behaviour of physical systems \citep{Gross_1996}.
Additionally, detection of new or broken symmetries is playing a role in the search for physics beyond the Standard Model \citep{atlas:2017:search}, particularly in data-driven approaches \citep{Karagiorgi_2022,Birman2022}. 
Recent work in machine learning and physics aims to learn or estimate symmetry groups from data \citep{Krippendorf_2020,zhou2021metalearning,Desai:2021,dehmamy2021automatic,Yang:2023aa} or to detect anomalous symmetry-breaking \citep{Collins2018anomaly,Birman2022}.
However, key inferential tools based on hypothesis tests for symmetry are missing.
Such tools are crucial if the discovery of symmetry from data is to be a reliable part of the scientific process. For example, they can be used to test for the presence or absence of a particular symmetry in data, with that symmetry specified by hypothesis or by a data-driven method that has learned or estimated a symmetry.
In situations with known or assumed symmetry, hypothesis tests for symmetry could also be used as model-checking criteria for models meant to exhibit that symmetry.

The present work formulates non-parametric tests, based on a single independent and identically distributed (i.i.d.) sample, for distributional symmetry under a specified group. 
We provide abstract formulations of tests that apply to two broad settings.
The first setting tests for the invariance of a marginal or joint distribution under the action of a compact group.
The test is formulated in such a way that if one has an asymptotically consistent estimator of a metric on the space of probability measures and the ability to sample uniformly random group elements, then it is straightforward to devise an asymptotically unbiased test for invariance.
More importantly, we design an easy-to-implement conditional Monte Carlo test that achieves exact $p$-values with finitely many observations and Monte Carlo samples. The test attains those properties by conditioning on a sufficient statistic induced by the group. 
The second setting tests for the invariance or equivariance of a conditional distribution under the action of a locally compact group, provided that the group action obeys weak regularity conditions.
We show that a test for conditional equivariance can be formulated as a particular test of conditional independence, which inherits the statistical properties of the conditional independence test chosen for implementation.
Although conditional independence testing is known to be hard \citep{Shah_2020}, especially as the dimension of the problem increases, the structure induced by symmetry means that the conditioning variables used in the test are often of much lower dimension than the observations. 

In addition to the generic testing methods and the study of their theoretical properties, we provide specific instantiations of the tests using kernel-based methods for non-parametric hypothesis testing.
We study these tests empirically on synthetic data, and apply them to geomagnetic satellite data and to two problems in high-energy particle physics.
Computer code required to run the experiments can be found on the \href{https://github.com/chiukenny/Tests-for-Distributional-Symmetry}{GitHub repository}\footnote{\url{https://github.com/chiukenny/Tests-for-Distributional-Symmetry}} for this work.

\subsection{Overview} \label{sec:overview}

The remainder of this section describes our work and main results informally, refraining for now from addressing technicalities.
The mathematical object that encodes symmetry is a group $\grp$.
Relevant technical details of groups are given in \cref{sec:background}.
Elements $g\in\grp$ act via transformations $x \mapsto g x$ of elements from a sample space $x \in \bfX$.
This action on $\bfX$ extends to the set $\calP(\bfX)$ of probability measures on $\bfX$.
If $P$ is the distribution of a random element $X \in \bfX$, then $g$ acts on $P$ via the pushforward, $g_* P(A) \define P(g^{-1}A)$, with $A \subseteq \bfX$ and $g^{-1}A \define \{ g^{-1}x : x \in A \}$.
A key question in many settings is whether the distribution $P$ underlying a set of i.i.d. observations $X_{1:n} \define (X_1,\dotsc,X_n)$ is \textit{invariant} under $\grp$ in the sense that
\begin{align*}
    g_* P = P \;, \quad \text{for each } g \in \grp \;.
\end{align*}
Outside of ill-behaved situations that typically do not arise in practice, this is only possible for a probability measure when $\grp$ is compact.
Any compact group $\grp$ has a unique invariant probability measure that can be thought of as the uniform distribution on $\grp$.
We denote this measure by $\haar$. 

For a specified group $\grp$, the statistical problem we address is to test the hypotheses 
\begin{align*}
    \nullHyp\colon P \text{ is $\grp$-invariant} \quad \text{versus} \quad \altHyp\colon P \text{ is not $\grp$-invariant} \;.
\end{align*}
If $\grp$ is relatively small and finite, or generated by a small set of elements (say of size $m$), invariance might be tested with a composite of $m$ two-sample hypothesis tests.
For large discrete groups, this approach quickly becomes untenable; for uncountable groups, it is not possible.
Instead, one might formulate hypothesis tests based on other characterizations of distributional invariance.
In \cref{thm:invariant_characterizations}, we collect a number of known identities that uniquely characterize the invariance of a probability measure and on which hypothesis tests may be based.
Perhaps the most well-known of the identities is that $P = \Pavg$ if and only if $P$ is $\grp$-invariant, where $\Pavg$ is obtained by averaging $g_* P$ over $\grp$ with respect to the invariant probability measure $\haar$,
\begin{align} \label{eq:intro:g:avg}
    \Pavg(A) \define \int_{\grp} P(g^{-1}A) \; \haar(dg) \;.
\end{align}
Because both $P$ and $\Pavg$ are probability measures on $\bfX$, any metric $\metric$ on $\calP(\bfX)$ can be used in conjunction with the empirical measure and a Monte Carlo estimate of the integral in \eqref{eq:intro:g:avg} to define a test statistic of the form
\begin{align*}
    \testStat_{n,m}(X_{1:n}) \define \metric\bigg( \frac{1}{n}\sum_{i=1}^n \delta_{X_i}(\argdot),\;\; \frac{1}{n m}\sum_{i=1}^n \sum_{j=1}^m \delta_{G_j X_i}(\argdot) \bigg) \;, \quad G_j \simiid \haar \;.
\end{align*}
This approach is very general and can be used for abstract spaces $\bfX$ other than $\R^d$ as long as one has access to a metric on $\calP(\bfX)$ and the ability to sample random elements of $\grp$ acting on $\bfX$.
An asymptotically consistent estimator of $\metric$ then yields an asymptotically unbiased test for invariance as $n,m\to\infty$ (\cref{thm:asymptotic:size:power}).

Beyond the general-purpose averaging approach, more detailed structure induced by $\grp$ is often available.
The group action partitions $\bfX$ into equivalence classes called \textit{orbits} so that $x$ and $x'$ are equivalent if and only if $x = gx'$ for some $g \in \grp$.
One can choose a representative element $\orbRep{x}$ of each orbit to define an \textit{orbit selector} $\orbSel(x) = \orbRep{x}$. 
The orbit selector allows one to decompose probability measures on $\bfX$, so that a random variable $X$ has an invariant distribution if and only if it satisfies $X \equdist G \orbSel(X)$, where $G \condind X$ is sampled uniformly from $\grp$.
As in the averaging test, one can then conduct a non-parametric test for invariance by comparing the untransformed empirical measure with the empirical measure of the observations $(G_1 \orbSel(X_1),\dotsc, G_n \orbSel(X_n))$. Such a test is asymptotically unbiased when based on a consistent estimator of a metric on $\calP(\bfX)$. 

More importantly, $\orbSel(X)$ is a special case of a \emph{maximal invariant} statistic, which is an invariant function that takes a different value on each orbit and thus uniquely encodes the orbits.
It is known that any maximal invariant is a sufficient statistic for $\invProbs$, the class of $\grp$-invariant probability distributions.
Sufficiency in particular means that for each $P\in\invProbs$, a sample $X_{1:n}\simiid P$ has the same conditional distribution given $\orbSel(X)_{1:n}$, and we can generate samples from that conditional distribution as $(G_1^{(b)}\orbSel(X_1),\dotsc,G_n^{(b)}\orbSel(X_n))$, with $G_i^{(b)}$ independent of $X_{1:n}$ and sampled i.i.d.\ from $\haar$.
The ability to sample from this shared conditional distribution enables us to design a conditional Monte Carlo test that yields exact $p$-values with finitely many observations and Monte Carlo samples.
\Cref{sec:inv:mc} details the test (\cref{alg:mc:pval}) and its statistical properties (\cref{thm:mc:test:p}). In \cref{sec:power:estimates}, we describe a method for estimating the conditional power function at the empirical measure of the observed data, $\empMeas{n}$, which can be combined with standard bootstrap resampling to estimate the power function at $P$. 

If $\grp$ acts freely on $\bfX$ in the sense that $gx = x$ implies $g$ is the identity element of $\grp$, then the orbit selector $\orbSel$ can be ``inverted'' to obtain the element of $\grp$ that sends $\orbRep{x}$ to $x$.
We call such a function, denoted $\repInv \colon \bfX \to \grp$, a \textit{representative inversion} because it satisfies $\repInv(x) \orbSel(x) = \repInv(x) \orbRep{x} = x$.
Yet another characterization of $\grp$-invariance is that for $X\sim P$, $\repInv(X) \equdist G$, with $G\condind X$ sampled uniformly from $\grp$.
In this case, a metric on the space of probability measures on $\grp$, $\probs[\grp]$, can be used as a test statistic.
\Cref{thm:mc:test:p} is easily adapted to this situation because $\haar$ is the unique invariant probability measure on $\grp$, and no appeal to sufficiency is required.
If the action of $\grp$ is not free, so that $gx = x$ for $g$ in some non-trivial subset $\grp_x \subseteq \grp$, then $\repInv$ can be replaced by an appropriate random variable $\repInvRand$ sampled from an \textit{inversion kernel}, $\repInvKern(x,\argdot)$.
The inversion kernel has a number of remarkable properties; the relevant one here is that if $\repInvRand \sim \repInvKern(x,\argdot)$, then $\repInvRand \orbSel(x) = x$ with probability one.
From this, it follows that a characterization of $\grp$-invariance is that $\repInvRand \equdist G$, with $G\sim\haar$.  

In each of the above cases, a non-parametric test for distributional invariance can be constructed by sampling random group transformations and applying them to a sample of data, then using an estimator for a metric on $\calP(\bfX)$ (or $\calP(\grp)$, as appropriate) to compare the untransformed sample with the randomly transformed sample. 
This recipe is generic and can in principle be used with any metric on $\calP(\bfX)$ (resp.~$\calP(\grp)$).
In \cref{sec:kernel:inv:coo}, we formulate specific versions of the tests using the kernel maximum mean discrepancy.
In \cref{sec:experiments}, we present the results of an empirical study of these tests on synthetic data to validate the theoretical properties obtained in \cref{thm:mc:test:p}, and apply the tests to two different applications: geomagnetic satellite data and simulated dijet events from the Large Hadron Collider.

\subsubsection{Conditional symmetry}

In some problems, especially those involving regression, classification, or prediction of a variable $Y \in \bfY$ from $X$, primary interest is in symmetry of the conditional distribution $P_{Y|X}$.
The conditional distribution is said to be \textit{equivariant} if for each measurable subset $B \subseteq \bfY$,
\begin{align*}
    P_{Y|X}(gx, B) = P_{Y|X}(x, g^{-1}B) \;, \quad x \in \bfX,\ g \in \grp \;.
\end{align*}
It is said to be invariant if the action of $\grp$ on $\bfY$ is trivial, so that the above equation holds with $g^{-1}B$ replaced by $B$.
Equivariant conditional distributions arise from the disintegration of jointly invariant probability distributions $P_{X,Y} = P_X \otimes P_{Y|X}$.
If $\grp$ is compact and the marginal distribution $P_X$ is known to be invariant, then testing for conditional equivariance of $P_{Y|X}$ is equivalent to testing for the joint invariance of $P_{X,Y}$, which could be carried out using the methods described above. 
However, the marginal distribution of $X$ may not be invariant---in many cases it is known not to be---and the problem cannot be reduced to a test for joint invariance. 
For example, if $\grp$ is non-compact, then $P_X$ cannot be $\grp$-invariant, but $P_{Y|X}$ may be. 
Instead, in \cref{thm:equivariance:cond:ind}, we obtain a conditional independence property that characterizes equivariance.
In particular, we show that $P_{Y|X}$ is equivariant if and only if 
\begin{align*}
    (\repInvRand, X) \condind \repInvRand^{-1} Y \mid \orbSel(X) \;, \quad \text{with} \quad \repInvRand \mid X \sim \repInvKern(X,\argdot) \;.
\end{align*}
Here, $\repInvRand^{-1}$ denotes the group inverse of the element $\repInvRand \in \grp$.
Moreover, $\orbSel(X)$ can be replaced by any {maximal invariant} statistic $\maxInv(X)$. 
This generalizes a related result of \citet{BloemReddy:2020}. 
A consequence of the result is that a test for conditional symmetry (equivariance or invariance) can be formulated as a test for conditional independence. %
In \cref{sec:kernel:equiv}, we describe an instantiation of this test using a kernel-based conditional independence test \citep{Zhang:2011}.
We apply it to synthetic data in \cref{sec:experiments}, as well as to data from three settings in the physical sciences.

\subsection{Related work}

There is an extensive literature on invariant testing problems; a textbook treatment can be found in \cite{Lehmann:Romano:2005}.
The main idea in invariant testing is as follows.
Let $\Omega$ represent the set of probability distributions under consideration, with $\Omega_0$ representing the subset that satisfy the null hypothesis and $\Omega_1$ those that do not, so that $\Omega = \Omega_0 \cup \Omega_1$ and $\Omega_0 \cap \Omega_1 = \emptyset$.
The testing problem is said to be $\grp$-invariant if each of $\Omega_0$ and $\Omega_1$ are $\grp$-invariant sets.
That is, if for every $P \in \Omega_0$, $g_* P \in \Omega_0$ for each $g \in \grp$, and likewise for every $P \in \Omega_1$.
There is extensive theory for such tests, much of which focuses on the optimality of tests based on maximal invariant statistics. 
A related line of work studies randomization tests in which random group elements are used for randomization; randomized permutation and sign-flip tests are two of the most common.
The literature on such tests for specific finite groups is substantial.
Recent work by \citet{Dobriban_2022} obtains consistency results for randomization tests using compact groups, under an additive noise decomposition assumption, and includes a review of earlier invariance-based randomization literature.  

The testing framework proposed here in \cref{sec:symmetries} trivially fits in the invariant testing framework, and the conditional Monte Carlo test proposed in \cref{sec:inv:mc} is closely related to invariance-based randomization techniques, but little follows directly from those connections. 
Whereas invariant testing and invariance-based randomization deal with accounting for or leveraging symmetry so that the problem of testing \emph{under} symmetry is simplified, the current work addresses the problem of testing \emph{for} symmetry; the two are conceptually somewhat orthogonal, though the mathematical techniques bear some similarities.
For example, a special case of our result on the size of the conditional Monte Carlo test (\cref{thm:mc:test:p}) was obtained for finite groups by \citet{Hemerik2018}, and our \cref{thm:mc:test:p} is perhaps of independent interest for invariance-based randomization tests. 

In light of the potential importance of a formal hypothesis testing framework for the presence or absence of distributional symmetry, it is somewhat surprising that with the exception of recent work \citep{Fraiman:2021,Christie:2022,Christie:2023}, such a framework and corresponding methods are largely absent from the literature.
As we describe below, those recent methods make strong assumptions that limit their applicability.
Our methods are broadly applicable; our experiments compare to the existing methods where possible and demonstrate some settings where those methods cannot be used. 

The recent work of \citet{Fraiman:2021} applied the Cram\'er--Wold (CW) theorem to formulate non-parametric tests for group invariance.
Those tests rely on the group $\grp$ being generated by a (small) finite set $\grp_0$ of transformations such that each element of $\grp$ can be written as a finite product $g = g_1\dotsb g_m$, where for each $j$, either $g_j \in \grp_0$ or $g_j^{-1} \in \grp_0$.
When this assumption holds, it can lead to a reduction in the computational complexity of the test.
On the other hand, the assumption can only be satisfied by discrete groups, as no uncountable group can be finitely generated.
Conversely, both the abstract formulation of our tests and our kernel-based implementations can be applied to any compact group, which includes finite discrete groups.
We compare both approaches empirically in \cref{sec:experiments} and find that although the computational complexity of the CW-based tests is favourable, they tend to be less powerful than the kernel-based tests we implemented.

To our knowledge, the test we propose in \cref{sec:equiv,sec:kernel:equiv} constitutes the first test (parametric or non-parametric) for symmetry of a conditional distribution.
In recent work, \citet{Christie:2023} proposed two tests for $\grp$-invariance of the conditional expectation $f(x) = \E[Y|X=x]$, $f \colon \bfX \to \R$.
Both of those tests require the user to assume that $f$ belongs to some specific class of functions, $\calF$, of bounded variation, and the assumption of an additive noise model, $Y_i = f(X_i) + \eps_i$, for independent mean-zero noise $\eps_i$.
One test requires knowledge of the bound $V(x,x') = \sup_{f\in\calF}|f(x) - f(x')|$ and a bound on the deviations on the noise variable, $\Pr(|\eps_i - \eps_j| \geq c) \leq p_c$.
The other test is less restrictive, instead requiring knowledge of some $\calV(x,x')$ satisfying $|f(x) - f(x')| \leq C_f \calV(x,x')$.
In our experiments in \cref{sec:experiments}, these assumptions are too restrictive for the tests to be applicable.
We note that the primary aim of \citet{Christie:2023} is to estimate the \emph{maximal} group under which $f$ is invariant, which amounts to conducting a collection of tests over a subgroup lattice of some candidate maximal group.
In principle, our tests could be substituted into their procedure, though we do not address that problem here. 

Apart from hypothesis testing, researchers in physics and machine learning have developed methods for estimating symmetries from data; see the references in \cref{sec:intro}. Hypothesis tests for symmetry, either as part of the estimation procedure or as validation of the estimated symmetry, have not been developed. To the best of our knowledge, the only exception is \citep{Birman2022}, which develops a test for anomaly detection, but requires restrictive distributional assumptions and approximations.

\section{Background: Groups, actions, and invariant measures} \label{sec:background}

Throughout, $\bfX$ denotes a topological space and $\bfS_{\bfX}$ its Borel $\sigma$-algebra, so that $(\bfX, \bfS_{\bfX})$ is a standard Borel measurable space.
When there is no chance of confusion, we will refer to $\bfX$ as a measurable space, and likewise for $\bfY$ and $\bfM$.
Let $P$ be a probability measure defined on $\bfX$.
For a random variable $X$ taking values in $\bfX$, we write $\E_P[X]$ for the expectation of $X$ with respect to $P$, and $X \sim P$ to denote a random variable sampled from $P$. 
For any measurable function $f$ on $\bfX$, let $f_*P$ denote the pushforward, or image measure, with $f_{*}P(A) = P(f^{-1}(A))$ for all measurable sets $A \in \bfS_{\bfX}$.
We use $\dirac_x$ to denote the Dirac measure at a point $x$.

\subsection{Groups}

Groups are central to the methods developed here, so we review some basic group theory.  
A group $\grp$ is a set with a binary operation~$\cdot$ that satisfies the associativity, identity, and inverse axioms. We denote the identity element by $\id$. 
For notational convenience, we write $g_1g_2=g_1\cdot g_2$ for $g_1,g_2\in\grp$. 
The group $\grp$ is said to be measurable if the group operations $g \mapsto g^{-1}$ and $(g_1,g_2) \mapsto g_1 g_2$ are $\bfS_{\grp}$-measurable, where $\bfS_{\grp}$ is a $\sigma$-algebra of subsets of $\grp$. 
We assume throughout that $\grp$ has a topology that is locally compact, second countable, and Hausdorff (lcscH), which makes the group operations continuous.
We may then take $\bfS_{\grp}$ as the Borel $\sigma$-algebra, making $\grp$ a standard Borel space. 

For $A\subseteq \grp$ and $g \in \grp$, we write $gA = \{ gh : h \in \grp\}$ and $Ag = \{ hg : h \in \grp \}$.
A measure $\nu$ on $\grp$ is said to be left-invariant if $\nu(gA) = \nu(A)$ for all $A \in \bfS_{\grp}$, and right-invariant if $\nu(Ag) = \nu(A)$.
When $\grp$ is lcscH, there exist left- and right-invariant $\sigma$-finite measures $\haar_{\grp}$ and $\rhaar_{\grp}$, respectively, that are unique up to scaling~\citep[Ch.~2.2]{Folland:2016}, known as left- and right-Haar measures.
When there is no chance of confusion, we use $\haar$ to denote left-Haar measure.
If $\grp$ is compact, then $\haar = \rhaar$, and the unique normalized Haar measure acts as the uniform probability measure over the group. 
We use $G$ to denote a random element of $\grp$; when $\grp$ is compact, $G\sim\haar$ denotes a random group element sampled from $\haar$.

\subsection{Group actions}

We briefly review the relevant aspects of groups acting on sets and special properties that are used in our work.
Many of the mathematical techniques have appeared in various statistical contexts, and a thorough treatment can be found in \cite{Eaton:1989,WIjsman:1990,Eaton:2007}.
Inversion kernels (described below) do not seem to have been used in statistics or machine learning, perhaps owing to their relatively recent appearance in probability \citep{Kallenberg2011:skew}.
However, special cases in which deterministic versions (called representative inversions below) exist have appeared in statistics and machine learning \citep{BloemReddy:2020,winter2022unsupervised}. 

A group $\grp$ acts measurably on a set $\bfX$ if the group action $\grpAct\colon \grp\times\bfX \to \bfX$ is measurable relative to $\bfS_{\grp} \otimes \bfS_{\bfX}$ and $\bfS_{\bfX}$.
All actions in this work are assumed to be continuous (and therefore measurable), and for convenience we simply say that $\grp$ acts on $\bfX$, writing $gx=\grpAct(g,x)$ as short-hand. 
For a set $A\subseteq\bfX$, the group acts as $gA=\{gx : x\in A\}$.
For fixed $x \in \bfX$, the stabilizer subgroup is $\grp_x = \{ g \in \grp : gx = x \}$.
The action is called \textit{free} or \textit{exact} if $gx = x$ implies that $g = \id$, in which case $\grp_x = \{\id\}$ for all $x \in \bfX$. 
The orbit of $x \in \bfX$ is the set $\orbit(x)=\{gx : g\in\grp\}$.
The orbits partition $\bfX$ into equivalence classes, where two points are equivalent if and only if they belong to the same orbit.
If $\bfX$ has only one orbit, then the action is said to be \textit{transitive}.
It is not hard to show that if $hx = x'$ for $x \neq x'$, then $h \grp_x h^{-1} = \grp_{x'}$.
That is, the stabilizer subgroups of the elements of an orbit are all conjugate. 

A function $f$ with domain $\bfX$ is invariant if it is constant on each orbit: $f(gx) = f(x)$, $x \in \bfX, g \in \grp$.
In general, an invariant function may take the same value on different orbits. 
A \textit{maximal invariant} is an invariant function $\maxInv \colon \bfX \to \bfM$ that takes a different value on each orbit, so that if $\maxInv(x) = \maxInv(x')$, then $x = gx'$ for some $g \in \grp$.
Maximal invariants arise as particularly useful statistics in problems with group symmetry because any invariant function $f$ can be written as $f(x) = k(\maxInv(x))$, for some function $k$.
Maximal invariants are typically not unique.
However, they are all isomorphic to the canonical projection onto the quotient space, $\proj \colon \bfX \to \bfX / \grp,\ x \mapsto \orbit(x)$.
Measurability issues can arise when $\grp$ is non-compact; we discuss these below. 

Invariance is a special case of a more general property. Suppose $\grp$ acts on $\bfX$ and on another set $\bfY$; the group action may be different on each. A function $f\colon \bfX \to \bfY$ is \emph{$\grp$-equivariant} if $f(gx) = gf(x)$, $x \in \bfX, g \in \grp$. These properties extend to measures. 

\begin{definition} \label{def:invariance}
    A probability measure $P$ on $\bfX$ is \textit{$\grp$-invariant} if $P(g^{-1}A) = P(A)$ for all $g\in\grp$, $A \in \bfS_{\bfX}$. 
\end{definition}

We write $g_* P(A) = P(g^{-1}A)$ as the pushforward of $P$ under the action of $g\in\grp$.
In that notation, $\grp$-invariance of $P$ entails $g_* P = P$ for all $g \in \grp$. 

We say that $P_{X,Y}$ is jointly $\grp$-invariant if it is invariant in the sense of \cref{def:invariance} extended to $\grp$ acting on $\bfX \times \bfY$.
In addition to joint invariance, we may define symmetry in the conditional distribution.

\begin{definition} \label{def:dist:equivariance}
    The conditional distribution of $Y$ given $X$ is said to be \emph{$\grp$-equivariant}\footnote{Some authors refer to \eqref{eqn:dist:equiv:def} as invariance; we use equivariance to avoid confusion with the invariance of marginal and joint distributions, and to be consistent with current usage, especially with respect to equivariant functions.} if
    \begin{equation} \label{eqn:dist:equiv:def}
        P_{Y|X}(x,B) = P_{Y|X}(gx,gB) \;, \quad x \in \bfX, \ B \in \bfS_{\bfY}, \ g \in \grp \;.
    \end{equation}
    If the action of $\grp$ on $\bfY$ is trivial and $P_{Y|X}$ satisfies \eqref{eqn:dist:equiv:def} so that $P_{Y|X}(gx,B) = P_{Y|X}(x,B)$, then the conditional distribution is said to be $\grp$-invariant.
\end{definition} 
Both of these definitions also apply to general measures (i.e., probability measures and conditional distributions can be replaced by measures and Markov kernels, respectively).

\subsubsection{Representatives and inversions} \label{sec:reps:invs}

Our work makes extensive use of special entities that are somewhat non-standard in the recent invariance-based statistics and machine learning literature, so we review them here. 
We can assign a particular element of each orbit as the \textit{orbit representative}. 
We write $\orbRep{x}$ as the representative on the orbit $\orbit(x)$.
That is, $\orbRep{x}=gx$ for some $g\in\grp$.
The structural properties described below do not depend on which element of the orbit is chosen as the representative.
All of the properties are relative to a particular choice, and a different choice would result in the same properties relative to that choice.
For a particular choice of representatives, the subset of $\bfX$ consisting of each orbit's representative is denoted by $\orbRep{\bfX}$.
Note that $\orbRep{\bfX} \cap \orbit(x)$ consists of a single point; namely, $\orbRep{x}$.
A function $\orbSel \colon \bfX \to \orbRep{\bfX}$ that maps elements of $\bfX$ onto their corresponding orbit representatives in $\orbRep{\bfX}$ is called an \textit{orbit selector}.
Note that any orbit selector is a maximal invariant by definition. 
Conversely, a maximal invariant defines a choice of orbit representatives if the value it takes on each orbit is an element of the orbit.
If $\orbRep{\bfX}$ is a measurable subset of $\bfX$ and $\orbSel$ is a measurable function relative to $\bfS_{\bfX}$ and $\bfS_{\bfX} \cap \orbRep{\bfX}$, then $\orbRep{\bfX}$ is called a \textit{measurable cross-section}. 

A function $\repInv\colon\bfX\to\grp$ is called a \textit{representative inversion} if $\grpAct(\repInv(x),\orbSel(x)) = \repInv(x)\orbSel(x) = x$ and $\repInv(gx) = g\repInv(x)$ for all $x \in \bfX, g \in \grp$. 
The role of $\repInv$ is to return the element of $\grp$ that must be applied to move $\orbRep{x}$ to $x$.
Conversely, the inverse element, $\repInv(x)^{-1}$, moves $x$ to $\orbRep{x}$.
In order for $\repInv$ to be uniquely defined, the group action must be free.
If it is not, an equivariant \textit{inversion probability kernel}, or inversion kernel for short, $\repInvKern \colon \bfX \times \bfS_{\grp} \to [0,1]$, can be used in place of $\repInv$, so that a sample from $\repInvKern(X,\argdot)$ will transform $\orbSel(X)$ into $X$ with probability one. %
That is, if $X \sim P$ and $\repInvRand \mid X \sim \repInvKern(X,\argdot)$, then $X = \repInvRand \orbSel(X)$ almost surely. 
At a high level, one may think of the inversion kernel $\repInvKern(x,\argdot)$ as the uniform distribution on the left coset $g \grp_{\orbSel(x)}$, where $g\orbSel(x) = x$. 
In the case of a free action, this simplifies to $\repInvKern(x,\argdot)= \delta_{\repInv(x)}$. 
In some cases, a representative inversion can still be defined when the action is not free (see \cref{expl:SOd:invariance}), in which case an equivalent inversion kernel can be defined as $\repInvKern'(x,B) \define \repInvKern(\orbSel(x),\repInv(x)^{-1}B)$.

\subsubsection{Proper group actions} \label{sec:proper:action}

In the analysis of probabilistic aspects of group actions, measurability issues can arise without regularity conditions.
Throughout, we will assume that the group action is \emph{proper}. 
That is, there exists a strictly positive measurable function $h \colon \bfX \to \R_+$ such that for each $x \in \bfX$, we have $\int_{\grp} h(gx) \haar(dg) < \infty$ \citep{Kallenberg2007}.\footnote{This definition of proper group action is a slightly weaker, non-topological version of the definition used in previous work in the statistics literature \citep[e.g.,][]{Eaton:1989,WIjsman:1990}, and only requires the exitence of Haar measure. The topological version is as follows: the map $(g,x) \mapsto (gx, x)$ is a proper map, i.e., the inverse image of each compact set in $\bfX\times\bfX$ is a compact set in $\grp \times \bfX$. That definition implies the one used here; see \cite{Kallenberg2007} for details.} 
This is a standard assumption in statistical applications of group theory \citep[e.g.,][]{Eaton:1989,WIjsman:1990,McCormack2023} and is satisfied in many settings of interest. 
A sufficient condition for proper group action is when $\grp$ is compact and acts continuously on $\bfX$, which is the setting for our tests for invariance in \cref{sec:symmetries}.
When $\grp$ is non-compact, a group action can fail to be proper if $\grp$ is ``too large'' for $\bfX$ in the sense that the stabilizer subgroups are non-compact. A class of non-compact group actions known to be proper are those of the isometry group of a Riemannian manifold. 
For the purposes of this work, we rely on the assumption of proper group actions to guarantee the existence of measurable orbit selectors and inversion kernels, which turn out to have extremely useful properties. 
We gather some of those properties in a proposition, which is a collection of existing results. 

To state it, let $\nu$ be any bounded measure on $(\bfX,\bfS_{\bfX})$ and let $\bar{\bfS}^{\nu}_{\bfX}$ be the completion of $\bfS_{\bfX}$ to include all subsets of $\nu$-null sets, and denote by $\bar{\nu}$ the extension of $\nu$ to $\bar{\bfS}^{\nu}_{\bfX}$ \citep[see, e.g.][Proposition~1.3.10]{Cinlar_2011}.
All statements of $\bar{\nu}$-measurability in the following proposition are with respect to $\bar{\bfS}^{\nu}_{\bfX}$, so that a set $A\subseteq\bfX$ is $\bar{\nu}$-measurable if $A \in \bar{\bfS}^{\nu}_{\bfX}$.
Moreover, a function defined by a particular property is $\bar{\nu}$-measurable if it is measurable in the usual sense with respect to $\bar{\bfS}^{\nu}_{\bfX}$, and if the defining property holds with the possible exception of a $\bar{\nu}$-null set. 
Clearly, such a function would also be $\bar{\rho}$-measurable for any measure $\rho \ll \nu$.

\begin{proposition} \label{prop:orbsel:inv}
    Let $\grp$ be a lcscH group acting continuously and properly on $\bfX$, and $\nu$ any bounded measure on $\bfX$.
    Then the following hold:
    \begin{enumerate}
        \item The canonical projection $\pi \colon \bfX \to \bfX/\grp$ is a measurable maximal invariant, and any measurable $\grp$-invariant function $f \colon \bfX \to \bfY$ can be written as $f = f^*\circ \pi$, for some measurable $f^* \colon \bfX/\grp \to \bfY$.

        \item There exists a $\bar{\nu}$-measurable orbit selector $\orbSel \colon \bfX \to \orbRep{\bfX}$, which is a maximal invariant statistic, and it induces a $\bar{\nu}$-measurable cross-section $\orbRep{\bfX} = \orbSel(\bfX)$. 

        \item For a fixed $\bar{\nu}$-measurable orbit selector $\orbSel$, there exists a unique $\bar{\nu}$-measurable inversion probability kernel $\repInvKern \colon \bfX \times \bfS_{\grp} \to [0,1]$ with the following properties:
            \begin{enumerate}
                \item $\repInvKern$ is $\grp$-equivariant: For all $g \in \grp, x \in \bfX, B \in \bfS_{\grp}$, $\repInvKern(gx,B) = \repInvKern(x,g^{-1} B)$. 

                \item For each $x \in \bfX$, $\repInvKern(\orbSel(x),\argdot)$ is normalized Haar measure on the stabilizer subgroup $\grp_{\orbSel(x)}$.

                \item For each $x \in \bfX$, if $\repInvRand \sim \repInvKern(x,\argdot)$, then $\repInvRand \orbSel(x) = x$ with probability one. 

                \item If there is a $\bar{\nu}$-measurable representative inversion $\repInv \colon \bfX \to \grp$ associated with $\orbSel$ such that it satisfies $\repInv(x)\orbSel(x) = x$  and $\repInv(gx) = g\repInv(x)$ for each $x \in \bfX, g \in \grp$, then $\repInvKern'(x,B) = \repInvKern(\orbSel(x),\repInv(x)^{-1}B)$ is an equivalent inversion kernel.
                In particular, this holds when the action of $\grp$ on $\bfX$ is free, in which case $\grp_{\orbSel(x)} = \{\id \}$ and the inversion kernel is $\delta_{\repInv(x)}$. 
            \end{enumerate}

    \end{enumerate}
\end{proposition}

The measurability of the canonical projection is a result from functional analysis; see \citet[][Theorem~5.4]{Eaton:1989} for an extended statement and references.
Items 2--3c follow directly from results of \citet{Kallenberg2011:skew,Kallenberg:2017} on the existence of universally measurable versions of $\orbSel$ and $\repInvKern$.
Item 3d follows from 3a and 3b. 

In the remainder of the paper, we assume that the action of $\grp$ on any space is continuous and proper; these conditions will be implicit in statements such as ``let $\grp$ be a group that acts on $\bfX$''.
In particular, measurable orbit selectors and inversion kernels exist under these assumptions.

\section{Testing for distributional invariance} \label{sec:symmetries}

In this section, we develop an abstract framework for non-parametric tests for distributional invariance under a specified group.
The tests are based on known characterizations of distributional invariance.
We briefly review that background here before developing the hypothesis testing framework in \cref{sec:test:stat:inv,sec:inv:mc}. 

We are interested in testing for the $\grp$-invariance of a \emph{probability} measure, which, under the assumption that $\grp$ acts properly on $\bfX$, requires $\grp$ to be compact.
For the remainder of this section, we assume that $\grp$ is compact; non-compact groups will arise in the treatment of conditional symmetry in \cref{sec:equiv}. 
For a specified compact group $\grp$, given a sample of data $(X_1,\dotsc,X_n)\simiid P$ from an unknown distribution $P$, decide between:
\begin{align*}
    \text{$\nullHyp \colon$ $P$ is $\grp$-invariant \quad versus \quad $\altHyp\colon$ $P$ is not $\grp$-invariant} \;.
\end{align*}
Recall that $\probs$ denotes the set of probability measures on $\bfX$. 
The set of $\grp$-invariant probability measures is $\invProbs$, and the non-invariant ones are $\ninvProbs$, so that $\invProbs \cup \ninvProbs = \probs$ and $\invProbs \cap \ninvProbs = \emptyset$.
For a specified $\grp$, we use $\nullHyp$ and $\invProbs$ interchangeably, and similarly for $\altHyp$ and $\ninvProbs$. 

Distributional invariance can be characterized in several known ways.
One well-known way to obtain an invariant distribution is to average over the group. 
In particular, when $\grp$ is compact, we can define the probability measure obtained by \textit{orbit-averaging $P$} as
\begin{equation*}
\Pavg(A) \define \int_{\grp} g_*P(A) \haar(dg) = \int_{\grp} P(g^{-1}A) \haar(dg) \;, \quad A \in \bfS_{\bfX} \;.
\end{equation*}
We refer to $\Pavg$ as the \textit{orbit-averaged distribution}. The averaging operator yields a useful characterization of invariant probability measures.
The following proposition lists additional characterizations that will be useful in developing hypothesis tests for invariance. 

\begin{proposition} \label{thm:invariant_characterizations}
Let $\grp$ be a compact group acting on $\bfX$ and $P$ a probability measure on $\bfX$.
Let $\orbSel$ be a measurable orbit selector and 
$\repInvKern$ a measurable inversion kernel. 
With $X \sim P$, the following are equivalent:
\begin{properties}[label=I\arabic*.,ref=I\arabic*]
    \setcounter{propertiesi}{-1}
    \item $P$ is $\grp$-invariant. \label{prp:inv:def}
    \item $P=\Pavg$. \label{prp:inv:orb:avg}
    \item If $G \sim \haar$ with $G \condind X$, then $X \equdist GX$. \label{prp:inv:random:G}
    \item If $G \sim \haar$ and $Y \sim \orbSel_{*}P$ with $G\condind Y$, then $X \equdist GY$.
    This holds even conditionally on $\orbSel(X)$. That is, $(\orbSel(X), X, G) \equdist (\orbSel(X), G\orbSel(X), G)$, which implies that ${X \mid \orbSel(X) \equdist G\orbSel(X) \mid \orbSel(X)}$.  \label{prp:inv:generate}
    \item If $\repInvRand \sim \repInvKern(X,\argdot)$ and $G \sim \haar$ with $\repInvRand \condind G \mid \orbSel(X)$, then $\repInvRand \equdist G$.
    If there exists a representative inversion $\repInv(x)$, then this holds with $\repInvRand$ replaced by $\repInv(X)H$, where $H \sim \haar_{\grp_{\orbSel(X)}}$. \label{prp:inv:rep:invers} 
\end{properties}
\end{proposition}

It follows from the invariance of the Haar measure that \cref{prp:inv:def,prp:inv:orb:avg} imply each other, which is straightforward to verify.
\cref{prp:inv:random:G} is just a reformulation of \cref{prp:inv:orb:avg} in terms of random variables.
These properties hold regardless of the existence of a measurable orbit selector and inversion kernel. 
Proving that \cref{prp:inv:def,prp:inv:generate} imply each other is only slightly more involved.
An accessible proof can be found in \citet[][Theorems 4.3--4.4]{Eaton:1989}; see also \citet[][Theorem 7.15]{Kallenberg:2017}. 
\cref{prp:inv:rep:invers} follows from \cref{prp:inv:generate} and the identity $x \equas \repInvRand\orbSel(x)$.

The following examples illustrate the main ideas.

\begin{figure}[!t]
\centering
\includegraphics[scale=0.45]{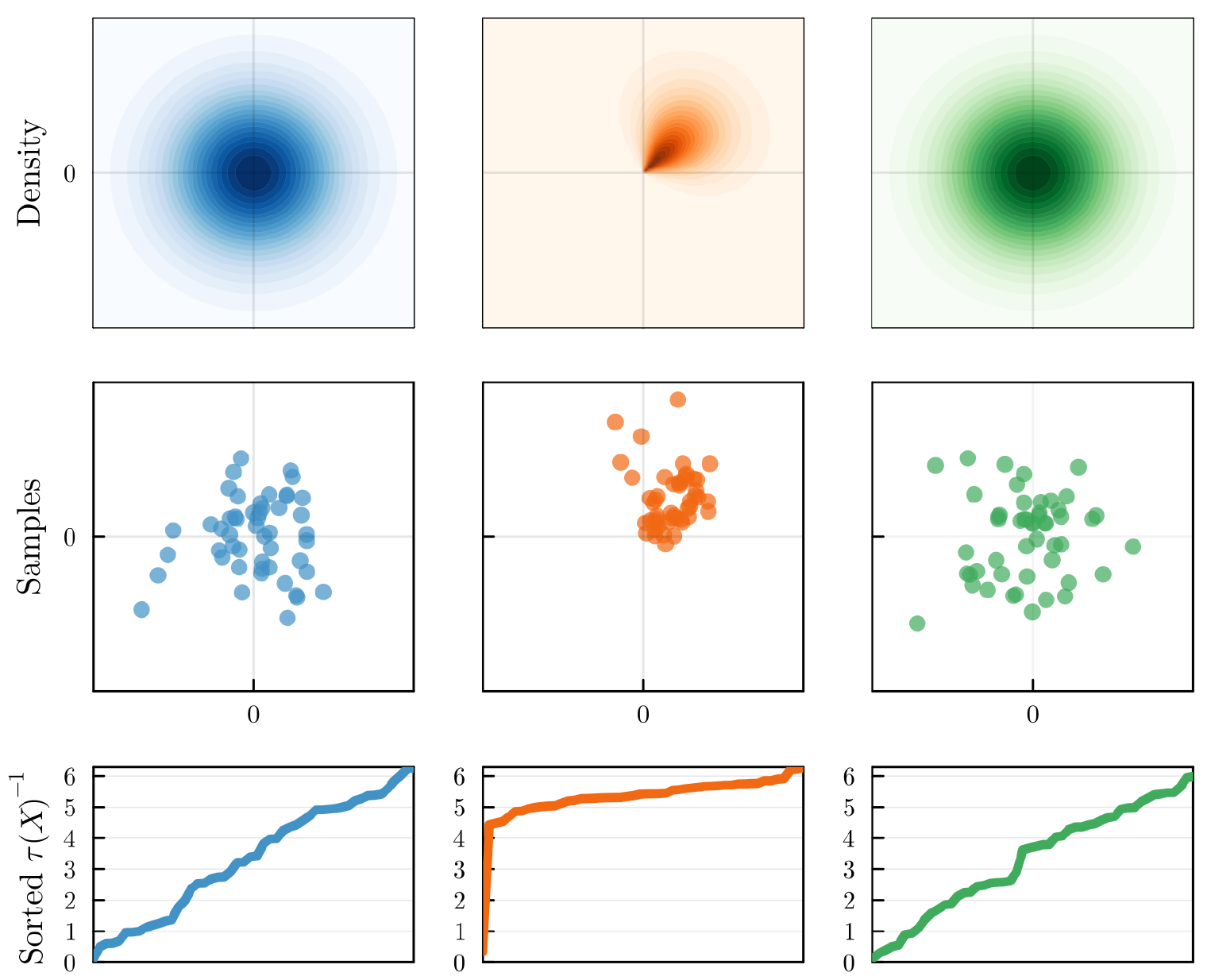}
\caption{First row: Densities for the 2D multivariate Gaussian $\gaussian(\zerovec_2,\I_2)$ (blue), Cartesian representation of the distribution $\chi_2$ $\otimes$ vonMises$(\nicefrac{\pi}{4},4)$ over polar coordinates (orange) and the same distribution averaged over $\SO{2}$ (green).
Second row: 50 samples from the respective distributions.
Third row: Angles in $[0,2\pi]$ needed for a counterclockwise rotation of each sample $X_i$ to the point $(\|X_i\|,0)$, sorted in increasing order.}
\label{fig:densities}
\end{figure}

\begin{example} \label{expl:SOd:invariance}
    Let $\bfX = \R^d$, so that $X$ is a random $d$-dimensional real vector.
    The isotropic multivariate normal distribution $\gaussian(0, \I_d)$ is a well-known example of a distribution that is invariant under the action of $\SO{d}$, the group of $d$-dimensional rotation matrices, where the action is by matrix-vector multiplication.
    \Cref{prp:inv:orb:avg,prp:inv:random:G} in \cref{thm:invariant_characterizations} are straightforward to check.
    Using the standard formula for an affine transformation of a multivariate normal distribution, if $X \sim \gaussian(0, \I_d)$ and $g \in \SO{d}$, then $gX$ has distribution $\gaussian(0, g\I_d g^{\top}) = \gaussian(0, \I_d)$.
    This holds for every $g$, and therefore it also holds for random $G$. 

    The other characterizations are perhaps less well-known.
    The set of orbit representatives (and the induced cross-section) can be chosen to be the points of an axis, e.g., the axis with unit basis vector $\e_1 = [1, 0, \dotsc, 0]^{\top}$.
    Then for each $x \in \R^d$, $\orbSel(x) = \|x\|\e_1$.
    For $d = 2$, the action is free; for $d > 2$, the stabilizer subgroup $\grp_{\orbSel(x)}$ is the set of $d$-dimensional rotations around the axis corresponding to $\e_1$.
    When $X \sim \gaussian(0, \I_d)$, $\|X\|$ has a so-called $\chi_d$-distribution (it is the square root of a $\chi^2_d$-distributed random variable), and $Y \equdist \|X\|\e_1$ satisfies $X \equdist GY$, with $G$ a uniform random rotation from $\SO{d}$.
    The left column of \cref{fig:densities} illustrates this for $d=2$.

    In this case, one may construct a representative inversion function corresponding to $\orbSel(x) = \|x\|\e_1$ by, for example, rotating $\|x\|\e_1$ to $x$ in the 2D subspace spanned by $\nicefrac{x}{\|x\|}$ and $\e_1$.
    That is, let $\tx \define \nicefrac{(x - \lip \e_1, x\rip \e_1)}{\| x - \lip \e_1, x\rip \e_1\|}$, so that $[\e_1,\ \tx]$ is a matrix in $\R^{d\times 2}$ whose columns form an orthonormal basis for the 2D subspace spanned by $\nicefrac{x}{\|x\|}$ and $\e_1$.
    Now let $\theta_x$ be such that $\cos(\theta_x) = \lip \e_1, \nicefrac{x}{\|x\|}\rip$, and $R_{\theta}$ the standard 2D rotation matrix of angle $\theta$,
    \begin{align*}
        R_{\theta} = 
        \begin{bmatrix}
            \cos(\theta) & -\sin(\theta) \\
            \sin(\theta) & \cos(\theta)
        \end{bmatrix} \;.
    \end{align*} 
    Then the $d$-dimensional rotation defined by
    \begin{align} \label{eq:SOd:rep:inv}
        \repInv(x) = \I_d - \e_1 \e_1^{\top} - \tx \tx^{\top} + [\e_1,\ \tx] R_{\theta_x} [\e_1,\ \tx]^{\top}
    \end{align}
    satisfies $\repInv(x)(\|x\|\e_1) = x$.
    A sample from the corresponding inversion kernel is then generated by taking a uniform random $(d-1)$-dimensional rotation $H$ and extending it to a $d$-dimensional rotation $H'$ that fixes the first dimension, so that $\repInv(x)H'$ has distribution $\repInvKern(x,\argdot)$.
    For $d=2$, $\repInvKern(\orbSel(x),\argdot) = \delta_{\id}$, so \cref{prp:inv:rep:invers} indicates that $\repInv(X) \equdist G$, with $G$ a uniform random 2D rotation.
    This is visualized in the bottom-left plot of \cref{fig:densities}.
\end{example}

\begin{example} \label{expl:SOd:noninvariance}
    As a non-invariant example, consider $\bfX = \R^d$ and $\grp = \SO{d}$ as in \cref{expl:SOd:invariance}, but now generate $X' = G'Y'$, with $Y' = Z\e_1$, $Z\sim \chi_d$, and $G'$ sampled from the von Mises--Fisher distribution $\vMF(\xi,\kappa)$, a non-invariant distribution on the $(d-1)$-sphere, which is isomorphic to $\SO{d}$.
    In this case, although the distribution of $Y'$ is the same as $Y = \|X\|\e_1$ above, rotations toward $\xi$ have higher probability. This is shown in the middle column of \cref{fig:densities}. 
    Averaging the $\vMF$ distribution over $\SO{d}$ results in the uniform distribution on $\SO{d}$, restoring invariance, as shown in the right column of \cref{fig:densities}. 
    This remains true even when the distribution of $G'$ depends on $Y'$, such as if $G'$ were sampled from $\vMF(\xi(y'),\kappa(y'))$. 
\end{example}

\begin{example} \label{expl:perm:invariance}
    As in the previous examples, let $\bfX = \R^d$.
    Now consider $\grp = \Sym{d}$, the symmetric group on $d$ elements acting on a vector $x \in \bfX$ by permutation.
    The distribution $P$ is said to be \emph{finitely exchangeable} if $gX \equdist X$ for every permutation $g \in \Sym{d}$.
    Here, the orbit representative is the vector of order statistics, $X_{(d)}$, which puts the elements of $X$ in increasing order.
    (We assume for convenience that ties occur with probability zero.)
    An exchangeable random variable can be generated by first sampling a random order statistic and then applying a uniform random permutation.
    The representative inversion $\repInv(X)$ is the permutation that transforms the order statistics back into $X$. 

    As an example, consider again the multivariate normal distribution, $\gaussian(0, \Sigma)$. 
    In order for this distribution to be exchangeable, the covariance matrix $\Sigma$ must satisfy \citep{Aldous:1983}
    \begin{align*}
        \Sigma = (1- \rho)\sigma^2\I_d + \rho\sigma^2\ones_d \;, \quad \sigma^2 > 0, \ -\frac{1}{d-1}\leq \rho \leq 1 \;,
    \end{align*}
    where $\ones_d$ is the $d\times d$ matrix of all ones.
    If, for example, $\Sigma_{1,d} = \Sigma_{d,1} = \rho^2\sigma^2$ for $|\rho|\neq 1$, then the resulting distribution would not be exchangeable. 
\end{example}

\subsection{Non-parametric test statistics} \label{sec:test:stat:inv}

The practical advantages of the equivalent characterizations of $\grp$-invariance will become clear in the following sections: rather than verifying potentially (uncountably) many equalities of the form $P = g_* P$, tests for invariance can be performed via a single comparison between a sample and random transformations of it. 
The distributional identities in \cref{thm:invariant_characterizations} suggest natural non-parametric test statistics in the form of divergences or metrics on the space of probability distributions on $\bfX$, so that any consistent estimator of such a statistic can be used to construct a consistent test of the desired asymptotic level. 

To that end, let $\metric \colon \probs \times \probs \to [0,\infty)$ be a metric on $\probs$.
(The following also holds if $\metric$ is a divergence or any other continuous function that separates points of $\probs$.)
The aim will be to recover $\metric(P,\Pavg)$ in the limit of infinite sample size.
Denote by $X_{1:n} \define (X_1,\dotsc,X_n)$ an i.i.d.\ sample from an unknown distribution $P$.
The samples can be used to estimate $P$ by the empirical measure
\begin{align*} %
    \empMeas{n}(A) \define \frac{1}{n}\sum_{i=1}^n\dirac_{X_i}(A) \;, \quad A \in \bfS_{\bfX} \;.
\end{align*}
Similarly, using i.i.d.\ samples $G_{i,j} \sim \haar$, an estimator of $\Pavg$ is the Monte Carlo averaged\footnote{We assume henceforth that $\grp$ is either uncountable or large and discrete, so that enumeration of the group elements is impossible or infeasible. If $\grp$ is discrete and small enough to feasibly enumerate, then Monte Carlo averages over $\grp$ can be replaced by exact averages.} empirical measure %
\begin{align*} %
    \empMeasAppAvg{n}{m}(A) \define \frac{1}{mn}\sum_{i=1}^n\sum_{j=1}^m \dirac_{X_i}(G_{i,j}^{-1} A) = \frac{1}{mn}\sum_{i=1}^n\sum_{j=1}^m \dirac_{G_{i,j} X_i}(A) \;, \quad A \in \bfS_{\bfX} \;.
\end{align*}
A natural test statistic is the distance between the two estimates,
\begin{align} \label{eq:metric:test:stat}
    \testStat_{n,m}(X_{1:n}) \define \metric(\empMeas{n}, \empMeasAppAvg{n}{m}) \;,
\end{align}
which converges almost surely to $\metric(P,\Pavg)$ as $n\to\infty$, for any $m \geq 1$. 

In practice, using $m=1$ amounts to a two-sample test between $(X_1,\dotsc,X_n)$ and $(G_1 X_1,\dotsc, G_n X_n)$.
For $m > 1$, it can be thought of as an aggregated $(m+1)$-sample test, where $m$ samples $(G_{1,j}X_1,\dotsc,G_{n,j} X_n)$ are known to be i.i.d.\ and therefore aggregated. 

For a sequence of metric-based statistics $(\testStat_{n,m})_{n\geq 1}$ with fixed $\metric$ and $m \geq 1$, and critical values $(c_n)_{n\geq 1}$, define the corresponding sequence of critical functions, or tests,
\begin{align} \label{eq:metric:tests}
    \critFun_{n,m}(X_{1:n}) = \indi\{ \testStat_{n,m}(X_{1:n}) > c_n \} \;.
\end{align}
The power function of a test based on $\critFun_{n,m}$ is
\begin{align*}
    \power_n(P) = \E_{P\otimes\haar}\!\left[\critFun_{n,m}(X_{1:n})\right] \;,
\end{align*}
where the expectation with respect to $P\otimes\haar$ is taken over $X_{1:n}$ and the random transformations $G_{i,j}\simiid\haar$.

Recall that $\invProbs$ is the set of $\grp$-invariant probability measures on $\bfX$ and $\ninvProbs$ the set of non-invariant probability measures so that $\probs = \invProbs \cup \ninvProbs$, and that the hypotheses are $H_0 \colon P \in \invProbs$ versus $H_1 \colon P \in \ninvProbs$. 

\begin{theorem} \label{thm:asymptotic:size:power}
    Fix $m \geq 1$ and a metric or divergence $\metric$ on $\probs$.
    Let a sequence of tests $(\critFun_{n,m})_{n\geq 1}$ (as in \eqref{eq:metric:tests}) be such that the critical values $(c_n)_{n\geq 1}$ satisfy $\lim_{n\to\infty} c_n = c \geq 0$.
    Then $(\critFun_{n,m})_{n\geq 1}$ is pointwise asymptotically level $\alpha$ for any $\alpha \in [0,1]$.
    That is, for any $c \geq 0$, for any $P \in \invProbs$, 
    \begin{align} \label{eq:asymp:level}
        \limsup_{n\to\infty} \E_{P\otimes\haar}\!\left[\critFun_{n,m}(X_{1:n})\right] \leq \alpha\;, \quad \alpha \in [0,1]\;.
    \end{align}
    If $c = 0$, then $(\critFun_{n,m})_{n\geq 1}$ is also pointwise consistent in power: for any $P \in \ninvProbs$,
    \begin{align} \label{eq:asymp:power}
        \lim_{n\to\infty} \!\E_{P\otimes\haar}\left[\critFun_{n,m}(X_{1:n})\right] = 1 \;,
    \end{align}
    and therefore the sequence of tests is asymptotically unbiased. 
\end{theorem}

The proof can be found in \cref{apx:proofs:size:power}. 
The theorem can be modified in a number of ways.
It remains valid if the metric $\metric$ is replaced by an estimator of the metric, $\hat{\metric}$, as long as $\hat{\metric}(\empMeas{n}, \empMeasAppAvg{n}{m})$ converges in probability to $\metric(P, \Pavg)$ as $n\to\infty$. 
The condition that $\metric$ distinguishes elements of $\probs$ (which is satisfied when $\metric$ is a metric or divergence) can be relaxed without changing the asymptotic level of the test.
All that is required is that $\metric(P,P) = 0$.
However, the power would be reduced if $\metric(P, \Pavg)= 0$ for $P \neq \Pavg$.

As an alternative, one may compare the distribution of inversions $\repInvRand_i \mid X_i \sim \repInvKern(X_i,\argdot)$ to Haar measure $\haar$ using a metric $\metric$ on $\probs[\grp]$, in which case an analogous version of \cref{thm:asymptotic:size:power} holds by \cref{thm:invariant_characterizations}, \cref{prp:inv:rep:invers}.

\subsection{Exact conditional Monte Carlo tests of invariance} \label{sec:inv:mc}

\Cref{thm:asymptotic:size:power} shows that tests based on a metric on $\probs$ have desirable large-sample properties, but the question of how to set critical values for finite $n$ remains unanswered.
In this section, we develop a conditional Monte Carlo method for doing so that results in a test with exact finite-sample size.
We do not address the more difficult theoretical question of power in this setting, though our experiments in \cref{sec:experiments} indicate that the methods do perform well in that respect. 

Our Monte Carlo procedure is conceptually similar to a resampling approximation of a two-sample permutation test for equality in distribution.
In the latter setting, under the null hypothesis that the two samples have the same distribution, a sufficient statistic for $\nullHyp$ is the empirical measure of the pooled sample; the conditional distribution of the pooled sample given the empirical measure is induced by the uniform distribution over permutations of the pooled sample.
Rather than computing the test statistic under every permutation, conditional $p$-values are estimated by sampling uniformly from the set of permutations. 
The conditional $p$-values are valid unconditionally since the conditional $p$-values are valid for almost every realization of the empirical measure under $\nullHyp$. 

Similarly, in the case of $\grp$-invariance, one may also condition on a sufficient statistic for $\invProbs$.
Here, any maximal invariant is a sufficient statistic \citep{farrell:1962,dawid:1985,BloemReddy:2020}.
For example, \cref{prp:inv:generate} of \cref{thm:invariant_characterizations} indicates that given $\orbSel(X)$, which is a maximal invariant, the conditional distribution of $X$ is that induced by $\haar$ on the orbit $\orbit(X)$.
That is, conditionally on $\orbSel(X)$, $X \equdist G \orbSel(X)$ with $G \sim \haar$.
\emph{This holds for every invariant probability measure $P \in \invProbs$.} 

Thus, the Monte Carlo testing method we propose in \cref{alg:mc:pval} generates pseudo-samples by first sampling $G^{(b)}_{i} \simiid \haar$ and then applying them to $\orbSel(X)_{1:n}$, so that
\begin{align} \label{eq:mc:test:sample}
    (G^{(b)}_1 \orbSel(X_1), \dotsc, G^{(b)}_n \orbSel(X_n)) \sim P(\argdot \mid \orbSel(X)_{1:n}) \;, \quad b = 1,\dotsc,B \;.
\end{align}
Because $P(\argdot \mid \orbSel(X)_{1:n})$ is the same for every $P \in \invProbs$, these samples can be used to estimate conditional quantities that are valid uniformly across the null hypothesis class $\invProbs$. 

If $\grp$ is discrete and relatively small, conditional expectations with respect to the right-hand side of \eqref{eq:mc:test:sample} can be computed exactly; otherwise, we can use Monte Carlo estimates.
In particular, given a sample $X_{1:n}$, we can estimate a conditional $p$-value by Monte Carlo sampling as in \cref{alg:mc:pval}.\footnote{Due to the invariance of $\haar$, $GX \equdist G\orbSel(X)$ (even conditioned on $\orbSel(X)$), so in practice we can replace $\orbSel(X_i)$ in \eqref{eq:mc:test:sample} with $X_i$.} 
\begin{algorithm}[h!]
\caption{Monte Carlo $p$-value} \label{alg:mc:pval}
\begin{algorithmic}[1]
\Procedure{McTest}{$X_{1:n},m,B,D$}
\State Sample $G_{j,1},\dotsc,G_{j,n}\simiid \haar$, for $j = 1,\dotsc,m$
\State Using $(G_{j,1},\dotsc,G_{j,n})_{j\leq m}$, compute $T_{n,m}(X_{1:n})$ as in \eqref{eq:metric:test:stat} 
\For{$b$ in $1,\dotsc,B$}
\State Sample $G^{(b)}_1,\dotsc,G^{(b)}_n \simiid \haar$
\State Set $X^{(b)}_{1:n} \define (G^{(b)}_1 X_1,\dotsc, G^{(b)}_n X_n)$ %
\State (Re)using $(G_{j,1},\dotsc,G_{j,n})_{j=1}^m$, compute $T_{n,m}(X^{(b)}_{1:n})$
\EndFor
\State \textbf{return} $p$-value $p_B$ computed as 
\begin{align} \label{eq:mc:test:p:value}
    p_B \define \frac{1 + \sum_{b=1}^B \indi\{T_{n,m}(X^{(b)}_{1:n}) \geq T_{n,m}(X_{1:n})\} }{1 + B}
\end{align}
\EndProcedure
\end{algorithmic}
\end{algorithm}
 
As we formalize below, this procedure produces a valid $p$-value for any $B \geq 1$.
The estimate $p_B$ can then be used in a critical function $\indi\{ p_B \leq \alpha \}$ 
and the resulting test has level $\alpha$. A special case of the following result, for finite $\grp$, appeared in \cite{Hemerik2018} in the context of invariance-based randomization tests. To state it, for $x \in \R_+$, let $\lfloor x \rfloor$ denote the ``floor'' function applied to $x$, i.e., the largest integer that is less than or equal to $x$.
Furthermore, let $X^{(0)}_{1:n} \define X_{1:n}$. 

\begin{theorem} \label{thm:mc:test:p}
    Assume that $\E_{P\otimes\haar}[\indi\{T_{n,m}(X^{(b)}_{1:n}) = T_{n,m}(X^{(b')}_{1:n})\}]=0$ for $b \neq b'$. 
    For any fixed $B \in \mathbb{N}$, $p_B$ obtained as in \cref{alg:mc:pval} is a valid $p$-value in the sense that for any $\alpha \in [0,1]$, if $P\in \invProbs$, then for any $(g_{i,j})_{i\leq n, j\leq m} \in \grp^{n\times m}$,
    \begin{align} \label{eq:mc:test:p:valid}
        \E_{P\otimes\haar}\left[\indi\{p_B \leq \alpha\} \mid (G_{i,j})_{i\leq n, j\leq m} = (g_{i,j})_{i\leq n, j\leq m}\right] = \frac{\lfloor \alpha (B+1) \rfloor}{B+1} \leq \alpha \;.
    \end{align}
    The same also holds unconditionally for random $(G^{(b)}_{i,j})_{i\leq n, j\leq m}$ sampled independently of $X_{1:n}$ such that they are exchangeable over the index $b=1,\dotsc,B$, which includes using the same random sample $(G_{i,j})_{i\leq n, j\leq m}$ for each $b$. 
\end{theorem}

The proof can be found in \cref{apx:proofs:p}. 
As noted by \citet{Dufour:2019:handbook}, if $\alpha (B+1)$ is an integer, then the inequality in \eqref{eq:mc:test:p:valid} becomes equality, and the critical region for the test, $\{p_B \leq \alpha\}$, has exact size $\alpha$.
We validate \cref{thm:mc:test:p} empirically in \cref{sec:experiments}, finding that over simulated datasets, the distribution of $p_B$ is approximately uniform under the null hypothesis and highly non-uniform under various alternatives. 

\cref{thm:mc:test:p} can be adapted in a few ways depending on the setting.
Firstly, if the probability of ties is non-zero (i.e., if $T_{n,m}$ is supported on a discrete set), then a modified version with randomized tie-breaking yields valid $p$-values, as described by \citet{Dufour:2006,Hemerik2018}.

Secondly, reusing $(G_{j,1},\dotsc,G_{j,n})_{j=1}^m$ is not strictly necessary in that \cref{thm:mc:test:p} still holds if a new set $(G^{(b)}_{j,1},\dotsc,G^{(b)}_{j,n})_{j=1}^m$ is sampled for each $b$. 
However, reusing the group elements in each iteration $b$ amounts to conditioning the $p$-value on them, as shown in \eqref{eq:mc:test:p:valid}. 
This conditioning reduces computation and potentially reduces estimation variance in the procedure.
The trade-offs between the two methods would appear primarily in the power of the test.
Furthermore, the $p$-value remains valid even if $(G_{j,1},\dotsc,G_{j,n})_{j=1}^m$ are not sampled from $\haar$. 
However, the power may suffer because averaging $P$ with respect to a probability measure other than $\haar$ does not result in an invariant distribution.
We find that reusing $(G_{j,1},\dotsc,G_{j,n})_{j=1}^m$ works well in our experiments in \cref{sec:experiments}. 

Thirdly, instead of reusing $X_{1:n}$ in each sampling iteration, we may combine the procedure with standard bootstrap resampling (sampling $X_{1:n}^{*(b)}$ i.i.d.\ with replacement from $X_{1:n}$) to obtain $X^{(b)}_{1:n} \define (G^{(b)}_1 X_1^{*(b)},\dotsc, G^{(b)}_n X_n^{*(b)})$.
We leave this possibility to future work. 

Finally, as with \cref{thm:asymptotic:size:power}, a version of \cref{thm:mc:test:p} holds for a suitably modified version of the Monte Carlo test that uses an observed sample of representative inversions, $(\repInvRand_i)_{i=1}^n$, where $\repInvRand_i \mid X_i \sim \repInvKern(X_i,\argdot)$. 
In that case, $X_i$ is replaced in the sampling procedure by $\repInvRand_i$, and the null hypothesis sample iterates $(G^{(b)}_1\repInvRand_1,\dotsc,G^{(b)}_n\repInvRand_n)$ would be compared to $(G_1,\dotsc,G_n)$ sampled i.i.d.\ from $\haar$, which is the unique invariant probability measure on $\grp$.

\subsubsection{Power estimates} \label{sec:power:estimates}

Given a sample $X_{1:n}$ and a fixed $\alpha \in [0,1]$, we may obtain an estimate of the power function at $\empMeas{n}$ of the test based on $\indi\{ p_B \leq \alpha \}$. 
Conditioned on $X_{1:n}$ and $G_{1:m,1:n} \define (G_{j,1},\dotsc,G_{j,n})_{j=1}^m$, the conditional power function at $\empMeas{n}$ is
\begin{align*}
    \power_{n,m}(\empMeas{n}, G_{1:m,1:n}) & = \E_{\haar}\!\left[\indi\{ p_B \leq \alpha \} \mid X_{1:n}, G_{1:m,1:n}\right]  \\
        & = \E_{\haar}\left[ \left.\sum_{b=1}^B \indi\left\{T_{n,m}(X^{(b)}_{1:n}) \geq T_{n,m}(X_{1:n}) \right\} \leq \alpha(B+1) - 1 \;\right| X_{1:n}, G_{1:m,1:n} \right] \;,
\end{align*}
where the expectation is taken over the conditional Monte Carlo samples. 
Because those samples are conditionally i.i.d., if $p_0 \define \E_{\haar}[\indi\{T_{n,m}(X^{(b)}_{1:n}) \leq T_{n,m}(X_{1:n}) \} \mid X_{1:n}, G_{1:m,1:n}]$, then
\begin{align} \label{eq:power:estimate}
    \power_{n,m}(\empMeas{n}, G_{1:m,1:n}) = \sum_{\ell = 0}^{\lfloor \alpha(B+1)-1\rfloor} \binom{B}{\ell} p_0^{\ell} (1-p_0)^{B-\ell} \;.
\end{align}
Estimates of $p_0$ can be obtained from \cref{alg:mc:pval} as $\hat{p}_0 = \nicefrac{(p_B(B+1) - 1)}{B}$. 

Alternatively, a new set of transformations $G^{(b)}_{1:m,1:n}$ can be sampled for each $b$.
The resulting Monte Carlo samples would still be conditionally i.i.d., so the resulting estimate of $p_0$ obtained in the same way as described above would be conditioned only on $X_{1:n}$.
An unconditional estimate of the power could be obtained using the usual bootstrap resampling methods for $X_{1:n}$, as shown in \cref{alg:power}.
We demonstrate this estimation procedure in some of our experiments in \cref{sec:experiments}.

\begin{algorithm}[h!]
\caption{Power estimate} \label{alg:power}
\begin{algorithmic}[1]
\Procedure{PowerEstimate}{$X_{1:n},m,C,B,D$}
\For{$c$ in $1,\dotsc,C$}
\State Sample $X^{(c)}_{1:n}$ i.i.d.\ with replacement from $\empMeas{n}$
\State Obtain $p^{(c)}_B$ from procedure \textsc{McTest} (\cref{alg:mc:pval}), either reusing $G_{1:m,1:n}$ or resampling for each $b$
\State Set $\hat{p}_0^{(c)} = \frac{p^{(c)}_B(B+1)-1}{B}$ and compute $\power_{n,m}^{(c)}$ as in \eqref{eq:power:estimate}
\EndFor
\State Set $\hat{\beta}_{n,m} \define \frac{1}{C}\sum_{c=1}^C \power_{n,m}^{(c)}$.
\EndProcedure
\end{algorithmic}
\end{algorithm}

\section{Kernel hypothesis tests for invariance} \label{sec:kernel:inv:coo}

As a demonstration of the abstract framework developed in \cref{sec:symmetries}, we develop the details of a particular instantiation using the maximum mean discrepancy (MMD) as the metric on $\probs$. 

Let $\hilbert$ be a reproducing kernel Hilbert space (RKHS) of functions $f\colon\bfX\rightarrow\R$ on which the evaluation functional $\funeval_x(f)=f(x)$ is continuous for all $x\in\bfX$.
For any $x\in\bfX$, the Riesz representation theorem says there exists a unique element $k_x\in\hilbert$ such that $f(x)=\lip f,k_x\rip_\hilbert$ for all $f\in\hilbert$, where $\lip\argdot,\argdot\rip_\hilbert$ is the inner product on $\hilbert$.
The reproducing kernel $k\colon\bfX\times\bfX\rightarrow\R$ of $\hilbert$ is a symmetric positive definite kernel such that $k_x(\argdot)=k(x,\argdot)$ and $k(x,x')=\lip k_x,k_{x'}\rip_\hilbert$.
The element $k_x$ can be viewed as a map $k_x\colon\bfX\rightarrow\hilbert$ that takes $x$ into a potentially infinite-dimensional feature space.
Kernel evaluations $k(x,x')=\lip k_x,k_{x'}\rip_\hilbert$ can then be interpreted as computing (possibly implicit) inner products on the mapped feature space. 
See \cite{steinwart:2008:svm} for a thorough treatment. 

The \textit{kernel mean embedding} (KME) of a distribution $P$ on $\bfX$ is defined as $\kme{P}\define\int_\bfX k_xP(dx)$ and is the unique element of $\hilbert$ such that $\E_P[f(X)]=\lip f,\kme{P}\rip_\hilbert$, for all $f\in\hilbert$ \citep{Muandet:2017}.
It follows that $\lip \kme{P_1},\kme{P_2}\rip = \int k(x,x') P_1(dx) P_2(dx')$.
We assume that the kernel $k$ is \textit{characteristic} so that the map $P\mapsto\kme{P}$ from $\probs$ into $\hilbert$ is injective, which leads to a unique embedding for each probability measure $P$ \citep{Sriperumbudur:2010}. 
If a non-characteristic kernel were used instead, the test statistic developed below would be unable to separate some distinct elements of $\probs$; as discussed in \cref{sec:test:stat:inv}, this would not affect the level of the test, but would reduce the power under non-invariant alternatives $P$ with $\kme{\Pavg} = \kme{P}$. 

Kernel-based hypothesis tests compare two distributions through their KMEs under a metric defined on functions in $\hilbert$, and can have an advantage over classical tests in that the same testing framework can be used for any type of data (e.g., vectors, matrices, images, etc.) as long as a kernel is available.
\citet{Gretton:2012} introduced a two-sample kernel hypothesis test based on the MMD.
The MMD is a particular integral probability metric defined as
\begin{align*}
    \MMD(P_1,P_2) \define \sup_{\|f\|_{\hilbert} \leq 1} \E_{P_1}\left[f(X)\right] - \E_{P_2}\left[f(X)\right] = \left\|\kme{P_1} - \kme{P_2}\right\|_\hilbert \;.
\end{align*}
We focus our attention on the \emph{squared} MMD, which is more practical to estimate due to the identity
\begin{align*}
     \MMD^2(P_1,P_2) & = \lip\kme{P_1},\kme{P_1}\rip_\hilbert + \lip\kme{P_2},\kme{P_2}\rip_\hilbert - 2\lip\kme{P_1},\kme{P_2}\rip_\hilbert \\
        & = \int k(x,x') P_1(dx) P_1(dx') + \int k(x,x') P_2(dx) P_2(dx') - 2 \int k(x,x') P_1(dx) P_2(dx') \;.
\end{align*}
This can be estimated from two samples of data, $X_{1:n_1}\simiid P_1$ and $Y_{1:n_2}\simiid P_2$, using the U-statistic
\begin{align*}
\widehat{\MMD}^2(\hat{P}_{1,n_1}, \hat{P}_{2,n_2}) = \frac{1}{n_1(n_1-1)}\sum_{i\neq j}k(X_i,X_j) + \frac{1}{n_2(n_2-1)}\sum_{i\neq j}k(Y_i,Y_j) - \frac{2}{n_1n_2}\sum_{i,j}k(X_i,Y_j) \;.
\end{align*}
The U-statistic is just an unbiased version of the empirical estimator, which is a V-statistic, 
\begin{align*}
\widehat{\MMD}^2_\text{V}(\hat{P}_{1,n_1}, \hat{P}_{2,n_2}) = \frac{1}{n_1^2}\sum_{i=1}^{n_1}\sum_{j=1}^{n_1}k(X_i,X_j) + \frac{1}{n_2^2}\sum_{i=1}^{n_2}\sum_{j=1}^{n_2}k(Y_i,Y_j) - \frac{2}{n_1n_2}\sum_{i=1}^{n_1}\sum_{j=1}^{n_2}k(X_i,Y_j) \;.
\end{align*}
For convenience, we refer to the MMD$^2$ as just the MMD, and similarly for related estimators.

In a standard two-sample test based on the MMD, the rejection region is estimated using a standard bootstrap procedure that involves pooling the two samples and repeatedly subsampling two pseudo-samples from the pool.
A test of level $\sig$ rejects $\nullHyp$ if the observed MMD estimate is larger than $(1-\sig)\times100\%$ of the bootstrapped values.
The two-sample MMD test is summarized in \cref{alg:mmd}. 

\begin{algorithm}[h!]
\caption{Two-sample MMD test} \label{alg:mmd}
\begin{algorithmic}[1]
\Procedure{2sMmdTest}{$X_{1:n_1},Y_{1:n_2},\sig$}
\State compute $\widehat{\MMD}(X_{1:n_1},Y_{1:n_2})$
\For{$b$ in $1,\dotsc,B$}
\State sample $X_{1:n_1}^{(b)}$, $Y_{1:n_2}^{(b)}$ from $X_{1:n_1}\cup Y_{1:n_2}$ with replacement
\State compute $\widehat{\MMD}(X_{1:n_1}^{(b)},Y_{1:n_2}^{(b)})$
\EndFor
\State compute $p$-value $p_B$ as in \cref{eq:mc:test:p:value} with $\widehat{\MMD}$ in place of $T_{n,m}$
\State \textbf{return} 1 if $p_B \leq \sig$ and 0 otherwise
\EndProcedure
\end{algorithmic}
\end{algorithm}

Computing $\widehat{\MMD}$ has a computational cost of $\bigO((n_1+n_2)^2)$, and the need to estimate the reference distribution via resampling means that the computation does not scale well with sample size.
The kernel literature has proposed cheaper approximations to the MMD, such as those based on random Fourier features and Nystr\"{o}m approximation \citep[e.g.,][]{Rahimi:2007,Raj:2017,Chatalic:2022}.
While these methods do not alleviate the need for the bootstrap, they make the computation closer to linear complexity, which consequently makes the test more feasible in the large sample regime.
In the following sections, we discuss how the two-sample MMD test can be repurposed as a test for invariance.
We only discuss the standard MMD formulation, but all MMD-based tests can also be reformulated as approximate tests based on random Fourier features (when the kernel is shift-invariant) or Nystr\"{o}m approximation.
In our experiments in \cref{sec:experiments}, we investigate the properties of a MMD test and its Nystr\"{o}m-approximated version.
We provide more details about the Nystr\"{o}m test statistic in \cref{apx:mmd:nystrom}.

\subsection{MMD test for invariance based on orbit-averaging} \label{sec:kernel:inv:orb:avg}

We first consider a test for $\grp$-invariance based on comparing $P$ and $\Pavg$ under the MMD.
The quantity of interest is
\begin{align*}
\MMD(P,\Pavg) &= \lip\kme{P},\kme{P}\rip_\hilbert + \lip\kme{\Pavg},\kme{\Pavg}\rip_\hilbert - 2\lip\kme{P},\kme{\Pavg}\rip_\hilbert \\
&= \lip\kme{P},\kme{P}\rip_\hilbert + \int_{\bfX\times\bfX}\int_{\grp\times\grp} k(gx,hx') \haar(dg)\haar(dh)P(dx)P(dx') \\
&\quad - 2\int_{\bfX\times\bfX} \int_{\grp} k(x, gx') \haar(dg) P(dx) P(dx') \;,
\end{align*}
which, given data $X_{1:n}\simiid P$ and sampled group actions $G_{1:m},H_{1:m}\simiid\haar$, is estimated by the test statistic 
\begin{align*}
    \widehat{\MMD}(\empMeas{n},\empMeasAppAvg{n}{m}) = \frac{1}{n(n-1)}\sum_{i\neq j}\left(k(X_i,X_j) + \frac{1}{m^2}\sum_{\ell=1}^m\sum_{r=1}^mk(G_\ell X_i,H_rX_j) - \frac{2}{m}\sum_{\ell=1}^mk(X_i,G_\ell X_j)\right) \;.
\end{align*}
To obtain a $p$-value in the test based on this statistic, we use the Monte Carlo sampling procedure described in \cref{sec:inv:mc}. 
If the kernel $k$ is (almost) equivariant in the sense that
\begin{equation} \label{eqn:equiv:kernel}
\int_\grp k(gx,x')\haar(dg) = \int_\grp k(x,gx')\haar(dg) \;,
\end{equation}
then a computationally more efficient estimator for the test statistic is possible.
We provide more details about the equivariant kernel setting in \cref{apx:mmd:equiv}.

\subsubsection{Baseline: two-sample MMD tests for invariance} \label{sec:kernel:inv:baseline}

Under $\nullHyp$, $g_i X_i \equdist X_i$ for each $g_i \in \grp$, $i=1,\dotsc,n$. 
Therefore, a standard two-sample test for equality in distribution (such as that in \cref{alg:mmd}) can be applied to the samples $X_{1:n}$ and $Y_{1:n} \define (g_1 X_1,\dotsc, g_n X_n)$. 
We can randomize the $g_i$'s and still have a test of the correct level.
We use this test, which we refer to as the \emph{transformation two-sample test}, as a sensible baseline since it is a valid test but does not take full advantage of the group structure via the sufficiency argument behind \cref{thm:mc:test:p}.

\subsection{MMD test for invariance based on representative inversions}

The tests for $\grp$-invariance described in \cref{sec:kernel:inv:orb:avg,sec:kernel:inv:baseline} can be modified to be based on representative inversions by replacing $X_{1:n}$ with their respective representative inversions $\repInv(X)_{1:n}$.
The test based on representative inversions can make use of a different resampling procedure where the new samples are directly sampled from $\haar$.

\subsection{Other non-parametric tests for invariance} \label{sec:other:tests}

While we focus on kernel-based approaches in this work, any test that involves estimating a metric defined on a space of probability distributions can be used.
Other possibilities include, for example, tests based on the Wasserstein distance or other integral probability metrics, or, if density evaluations are available, $f$-divergences such as the Kullback--Leibler divergence.
Among well-known integral probability metrics, the MMD has favorable computational and statistical rates \citep{Sriperumbudur_2012}. 

Alternatively, any property that uniquely characterizes a distribution can also be used to test for $\grp$-invariance.
Recent work by \citet{Fraiman:2021} employs the Cram\'{e}r--Wold (CW) theorem, which states that any two distributions on $\R^d$ are the same if and only if the distributions of their 1D projections via the linear kernel $x \mapsto x^\T t$ are the same, for all projections $t$.
To our knowledge, the resulting CW test is the only existing test for distributional group invariance. 

The CW test procedure is as follows.
For i.i.d.\ random variables $Z_{1:n}$ supported on $\R$, let $\widehat{F}_{Z_{1:n}}$ denote the empirical cumulative distribution function of a random variable $Z$.
The procedure proposed by \citet{Fraiman:2021} requires that $\grp$ be finitely generated by a subset of group elements of size $L$. 
Given data $X_{1:n}$ on $\R^d$, the group generators $(g_{\ell})_{\ell=1}^L$ and $J$ random unit vectors $t_j\in\R^d$ are used to compute the worst-case Kolmogorov--Smirnov statistic, 
\begin{align*}
\testStat_\text{CW}(X_{1:n}) = \max_{\substack{\ell\in 1:L\\j\in 1:J}} \; \sup_{u\in\R} \left|\widehat{F}_{\left(t_j^\T X\right)_{1:n}}(u) - \widehat{F}_{\left(t_j^\T (g_{\ell} X)\right)_{1:n}}(u)\right| \;.
\end{align*}
The $p$-value is estimated by standard bootstrap resampling from $X_{1:n}$.\footnote{There is also a version that does not rely on bootstrapping but requires the sample to be split and the use of a Bonferroni correction, which likely reduces the power.}
It is straightforward to extend the CW test to more general groups by sampling $G_{\ell}\simiid \haar$ and applying the methods of \cref{sec:symmetries} to obtain a valid test.
We compare our tests for invariance to the extended CW test in the experiments in \cref{sec:experiments}.

\section{Equivariance: conditional symmetry} \label{sec:equiv}

Although tests for marginal or joint distributional invariance are useful in a number of settings, many problems require a test for symmetry of a \emph{conditional} distribution.
These commonly arise in prediction settings where, for example, one must predict $Y$ from $X$, and it is of interest to test whether a change $X \mapsto gX$ corresponds to a change $Y \mapsto gY$.
These so-called {equivariant prediction} problems have been a subject of intense study in machine learning \citep[see][for a review]{Bronstein:2021aa}. 

Suppose our data is of the form $(X,Y)_{1:n}$, with each $(X,Y)$-pair sampled i.i.d.\ from some distribution $P_{X,Y}$ defined over the product space $\bfX\times\bfY$.
We write the disintegration of $P_{X,Y}$ as $P_{X,Y}=P_X \otimes P_{Y|X}$, where $P_{Y|X}$ denotes a regular conditional probability of $Y$ given $X$ (i.e., represented by a Markov probability kernel from $\bfX$ to $\bfY$).
Furthermore, suppose there is a group $\grp$ that acts on both $\bfX$ and $\bfY$; the group action may be different on each.
Recall that we say that $P_{X,Y}$ is jointly $\grp$-invariant if it is invariant in the sense of \cref{def:invariance} extended to $\grp$ acting on $\bfX \times \bfY$.

We refer to $\grp$-equivariance and $\grp$-invariance of the conditional distribution (as in \cref{def:dist:equivariance}) collectively as \textit{conditional symmetry}.
Subsequently, we describe ideas generally in terms of equivariance, which specializes to invariance when $\grp$ acts trivially on $\bfY$.
In contrast to \cref{sec:symmetries}, the symmetry here can be viewed as symmetry of a \emph{function on $\bfX$}, rather than of a probability measure on $(\bfX,\bfS_{\bfX})$. 
To see this, let $f \colon \bfX \to \bfY$ be $\grp$-equivariant, so that $f(gx) = g f(x)$, for all $x \in \bfX$, $g \in \grp$.
The conditional distribution of $Y \define f(X)$ can be represented by the Dirac kernel $\delta_{f(x)}(B)$ in this case.
It is not hard to show that this kernel is equivariant in the sense of \eqref{eqn:dist:equiv:def} if and only if $f$ is an equivariant function.
In the more general case, conditional symmetry describes how the conditional distribution transforms when $x$ is transformed to $gx$.
Indeed, an equivalent definition to \eqref{eqn:dist:equiv:def} is
\begin{align*}
    P_{Y|X}(gx,B) = P_{Y|X}(x,g^{-1}B) \;, \quad x \in \bfX, \ B \in \bfS_{\bfY}, \ g \in \grp \;.
\end{align*}
In this way, an equivariant conditional distribution ``transmits'' the transformation of one variable to the other.
In contrast to the invariance of a marginal (or joint) distribution, which must preserve the total probability mass, the transformations of a conditional distribution have no such restrictions and conditional symmetry is well-defined even when $\grp$ is non-compact. 

Equivariant Markov kernels occur naturally in the disintegration of jointly invariant measures.
Under quite general conditions, a measure is jointly invariant if and only if it disintegrates into an invariant marginal measure and an equivariant Markov kernel.
A measure-theoretic proof of this can be found in \citet[][Theorem 7.6]{Kallenberg:2017}.
\citet{BloemReddy:2020} proved a version of this result for probability measures and compact groups using different methods that require the existence of a measurable representative inversion $\repInv$, and which results in a detailed description of the probabilistic structure of the disintegration.
Those authors found that if $P_{X,Y}$ is \textit{jointly} $\grp$-invariant, then 
\begin{align} \label{eqn:cond:ind}
    X \condind \repInv(X)^{-1} Y \mid \maxInv(X) \;.
\end{align}
In many applications, we are interested in testing for conditional equivariance of $Y \mid X$, even when the marginal distribution of $X$ is not invariant. %
As we prove below, the conditional independence in \eqref{eqn:cond:ind} holds even when $P_{X,Y}$ is not jointly invariant.
In fact, equivariance of $P_{Y|X}$ is both sufficient and necessary.
Moreover, even when a unique inversion $\repInv$ does not exist---i.e., when the action of $\grp$ on $\bfX$ is not free---an analogous conditional independence relation applies: $P_{Y|X}$ is $\grp$-equivariant if and only if
\begin{align} \label{eqn:cond:ind:rand}
    (\repInvRand, X) \condind \repInvRand^{-1} Y \mid \maxInv(X) \;, \quad \text{with} \quad \repInvRand \mid X \sim \repInvKern(X,\argdot) \;.
\end{align}

\begin{theorem} \label{thm:equivariance:cond:ind}
    Let $\grp$ be a lcscH group acting on each of $\bfX$ and $\bfY$, with the action on $\bfX$ proper, so that a measurable inversion kernel $\repInvKern$ exists.
    Then $P_{Y|X}$ is conditionally $\grp$-equivariant if and only if \eqref{eqn:cond:ind:rand} holds.
    If there exists a measurable inversion function $\repInv \colon \bfX \to \grp$, then the characterization condition \eqref{eqn:cond:ind:rand} reduces to \eqref{eqn:cond:ind}.
    If the action of $\grp$ on $\bfY$ is trivial, then \eqref{eqn:cond:ind:rand} reduces to $X \condind Y \mid \maxInv(X)$. 
\end{theorem}

The proof can be found in \cref{apx:proofs:equivariance}. 
Based on \cref{thm:equivariance:cond:ind}, a test for $\grp$-equivariance of a conditional distribution can be performed as a test of the conditional independence in \eqref{eqn:cond:ind:rand}. 
If the action of $\grp$ on $\bfX$ is transitive, then there is a single orbit and the test simplifies to a test for (unconditional) independence. 

Any maximal invariant can be used as the conditioning variable.
A particularly convenient maximal invariant for testing purposes is the orbit selector $\orbSel(X)$.
Although all maximal invariants are isomorphic to one another, there may be computational and statistical benefits to using a maximal invariant with an explicit representation in a low-dimensional space.
For example, in the case of $\SO{d}$, one may use $\maxInv(X) = \|X\|$. 

\begin{example} \label{expl:SOd:equivariance}
    For random vectors $X$ and $Y$ in $\R^d$, suppose that $Y \mid X \sim \gaussian(X, \I_d)$.
    It is straighforward to check that $P_{Y|X}$ is $\SO{d}$-equivariant:
    If $\eps \sim \gaussian(0, \I_d)$ so that $Y = X + \eps$, then for any $g \in \SO{d}$,
    \begin{align*}
        gY \mid X = gX + g\eps \equdist gX + \eps = Y \mid gX \;.
    \end{align*}
    Using the representative inversion \eqref{eq:SOd:rep:inv} described in \cref{expl:SOd:invariance}, we can test for equivariance with the conditional independence test \eqref{eqn:cond:ind}.
    We do so in \cref{sec:experiments:synthetic}. 
\end{example}

\begin{example} \label{expl:lorentz}
    The Lorentz group, or indefinite orthogonal group $\lorentz$, is a fundamental symmetry group in physics, where it is the group of linear isometries of Minkowski spacetime, on which the theory of special relativity is commonly constructed. 
    It preserves the quadratic form
    \begin{align*}
        Q(p) = E^2 - p_x^2 - p_y^2 - p_z^2 \;,
    \end{align*}
    where $p = (E, p_x, p_y, p_z)$ is a vector in $\R^4$ representing a particle's momentum in Minkowski spacetime, known as the \emph{four-momentum}. Mathematically, the Lorentz group is a non-compact non-Abelian real Lie group that is not connected. 
    The Standard Model of physics is invariant under the Lorentz group, and violations of $\lorentz$-invariance could indicate phenomena beyond the Standard Model. An example application is quark tagging \citep{Kasieczka:2019,bogatskiy20a}. High-energy quarks produced in particle collisions, such as those at the Large Hadron Collider, quickly decay through a cascading process of gluon emission into collimated sprays of stable hadrons, which are subatomic particles that can be detected and measured \citep{Salam2010}.  This spray is known as a \textit{jet}. Identifying, or tagging, which species of quark gave rise to a jet is a crucial task in collider physics, and an active area of research \citep{LARKOSKI20201}. 
    
    According to the Standard Model, the inertial frame of the parent quark may differ from the lab frame by a Lorentz group transformation, and the task of quark-tagging based on the four-momenta of constituent particles should be invariant to those transformations. That symmetry was incorporated into a neural network architecture designed for quark-tagging \citep{bogatskiy20a}. Alternatively, one may wish to test for symmetry in a dataset of jet-quark pairs. In the language of conditional symmetry, given a collection of four-momenta of $\ell$ decay particles ($\bfX = \R^{\ell \times 4}$), the conditional probability that they decayed from a particular particle should be $\lorentz$-invariant.
    In our experiments in \cref{sec:experiments:top:quark}, interest is in whether or not the particles decayed from a top quark, so $\bfY = \{0, 1\}$. 
    Since $Q(p)$ is a maximal invariant, it can be used as the conditioning variable.
    Moreover, since conditional invariance is being tested (the action of $\lorentz$ on $\bfY$ is trivial), the conditional independence test based on $X \condind Y \mid \maxInv(X)$ can be used. 
\end{example}

\section{Kernel conditional independence test for equivariance} \label{sec:kernel:equiv}

\Cref{thm:equivariance:cond:ind} establishes that a test for $\grp$-equivariance of a conditional distribution can be formulated as a test for conditional independence of the form \eqref{eqn:cond:ind:rand}. 
Any non-parametric conditional independence test could be used for this purpose; see \citet{Li:2020} for a recent review of the literature. 
We use the kernel conditional independence (KCI) test \citep{Zhang:2011} in our experiments in \cref{sec:experiments}.

In all of our experiments, either a representative inversion $\repInv$ exists or the action of $\grp$ on $\bfY$ is trivial, so for simplicity we formulate the statistic for the simpler conditional independence test \eqref{eqn:cond:ind}.
Modification to accommodate random $\repInvRand$ is straightforward, by replacing $\repInv^{-1}(X_i)$ in the expression below with a sampled $\repInvRand_i \sim \repInvKern(X_i,\argdot)$ and extending the kernel $\bfK_{XM}$ as defined below on $\bfX \times \bfM$ to a kernel on $\grp\times\bfX\times\bfM$. 

The test statistic in the KCI test for equivariance is constructed as follows.
Let $k_X$, $k_Y$ and $k_\maxInv$ be kernels on $\bfX$, $\bfY$, and $\bfM$, respectively.
Given data $(X,Y)_{1:n}$, define the kernel matrices $\bfK_Y$, $\bfK_\maxInv$ and $\bfK_{X\maxInv}$ as
\begin{align*}
\left[\bfK_Y\right]_{ij} &= k_Y(\repInv(X_i)^{-1}Y_i,\repInv(X_j)^{-1}Y_j) \;, &
\left[\bfK_\maxInv\right]_{ij} &= k_\maxInv(\maxInv(X_i),\maxInv(X_j)) \;, &
\left[\bfK_{X\maxInv}\right]_{ij} &= k_X(X_i,X_j)\left[\bfK_\maxInv\right]_{ij} \;.
\end{align*}
Let $\bar{\bfK}_Y=\bfH\bfK_Y\bfH$ denote the centralized kernel matrix, where $\bfH=\I_n - n^{-1}\ones_n$, and similarly for $\bar{\bfK}_\maxInv$ and $\bar{\bfK}_{X\maxInv}$.
For fixed $\eps>0$, define the matrices $\bfR_\maxInv=\eps(\bar{\bfK}_\maxInv+\eps\I_n)^{-1}$, $\bar{\bfK}_{X\maxInv|\maxInv}=\bfR_\maxInv\bar{\bfK}_{X\maxInv}\bfR_\maxInv$, and $\bar{\bfK}_{Y|\maxInv}=\bfR_\maxInv\bar{\bfK}_{Y}\bfR_\maxInv$.
Then the test statistic is given by
\begin{align*}
\testStat_\text{KCI}(X_{1:n},Y_{1:n}) = \frac{1}{n}\mathrm{Tr}(\bar{\bfK}_{X\maxInv|\maxInv}\bar{\bfK}_{Y|\maxInv}) \;.
\end{align*}
The distribution of this test statistic under $\nullHyp$ can be approximated by samples $T^{(1)},\dotsc,T^{(B)}$ drawn through a simulation procedure described by \citet{Zhang:2011}, and the test rejects $\nullHyp$ at level $\sig$ if
\begin{align*}
\frac{1}{B}\sum_{b=1}^B\indi\left\{\testStat_\text{KCI}(X_{1:n},Y_{1:n}) \leq \testStat^{(b)}\right\} \leq \sig \;.
\end{align*}
See \citet{Zhang:2011} for more details and explanations.

As a point of comparison, we also implement a test for equivariance based on the conditional permutation (CP) test \citep{Berrett:2019} where kernel conditional density estimation (KCDE, \citealt{Gooijer:2003}) is used in the sampling procedure.
We leave the details of this particular test to \cref{apx:cpt}.

\section{Experiments} \label{sec:experiments}

We evaluate our proposed tests on several synthetic and real data examples.
In all of our experiments, we sample $n$ data points from a distribution or dataset and perform a test for a specified symmetry.
We repeat this procedure over $N=1000$ simulations for each test and record the proportion of simulations in which the test rejected.
The proportion of rejections is an estimate of the level of the test when the data distribution has the specified symmetry and otherwise an estimate of the power when the distribution does not.
With $N=1000$ simulations, estimates are precise up to approximately $\pm0.016$.
We set the desired significance level to be $\sig=0.05$ in all of our experiments.
We use $m=2$ sampled group actions except where otherwise specified.
We default to Gaussian radial basis function kernels for data that are continuous and use the median distance heuristic \citep{Garreau:2017} for the kernel bandwidth unless otherwise specified.
The median distance is computed from a ``training'' set of $n$ data points randomly split from the ``test'' set used to estimate the rejection rate.
The median distance is recomputed in every simulation.

Let $\lceil x\rceil$ denote the ``ceiling'' function.
The tests that we evaluate for invariance include: the transformation two-sample test based on the standard two-sample MMD (\tSMMD, \cref{sec:kernel:inv:baseline}), the MMD test (\tMMD, \cref{sec:kernel:inv:orb:avg}), the MMD test with Nystr\"{o}m approximation using $J=\lceil\sqrt{n}\rceil$ subsamples (\tNMMD, \cref{apx:mmd:nystrom}), and the CW test with $J=\lceil\sqrt{n}\rceil$ random projections (\tCW, \cref{sec:other:tests}).
Where applicable, we use the sampling procedure described in \cref{alg:mc:pval} with $B=200$.
We test for conditional symmetries using the KCI test (\tKCI, \cref{sec:kernel:equiv}), and the CP test with KCDE (\tCP, \cref{apx:cpt}, $S=50$ steps) using the multiple correlation coefficient \citep{Abdi:2007} as the test statistic. 
In most conditional symmetry experiments, we find that using the median heuristic for \tKCI and \tCP does not lead to meaningful results.
In these cases, we select the kernel bandwidths by performing a grid search for each kernel.
For each combination of bandwidths, we estimate the size and power of the test over 100 simulations involving training data.
We choose the combination that leads to a rejection rate of at most $0.1$ on data generated under $\nullHyp$ and that maximizes rejection rate on data generated under $\altHyp$.
If no combination has rejection rate less than $0.1$ under $\nullHyp$, we then use the combination that leads to the lowest rejection rate. 
Further details about the grid search for each experiment, along with other experiment details, can be found in \cref{apx:exp}.

All experiments were implemented in the Julia programming language (version 1.6.1) and run with a single thread on a high-performance computing allocation with 4 CPU cores and 16GB of RAM.

\begin{table}[!t]
\begin{tabular*}{\columnwidth}{@{\extracolsep\fill}rrrrrrrrrrrr@{\extracolsep\fill}}
\toprule
& \multicolumn{5}{c}{$\grp = \SO{4}$} & \multicolumn{6}{c}{$\grp = \Sym{10}$} \\
\cmidrule(lr){2-6}
\cmidrule(l){7-12}
& \multicolumn{3}{c}{Invariance} & \multicolumn{2}{c}{Equivariance} & \multicolumn{4}{c}{Invariance} &  \multicolumn{2}{c}{Equivariance} \\
& $\nullHyp$ & $\altHyp$ & $\hat{\power}_{n,m}$ & $\nullHyp$ & $\altHyp$ & $\nullHyp^+$ & $\nullHyp^-$ & $\altHyp$ & $\hat{\power}_{n,m}$ & $\nullHyp$ & $\altHyp$ \\
\midrule
\tSMMD & $0.041$ & $0.870$ & $0.778$ &&& $0.012$ & $0.052$ & $0.987$ & $1.000$ \\
\tMMD & $0.050$ & $0.984$ & $0.961$ &&& $0.047$ & $0.053$ & $1.000$ & $1.000$ \\
\tNMMD & $0.051$ & $0.896$ & $0.743$ &&& $0.054$ & $0.044$ & $0.122$ & $0.175$ \\
\tCW & $0.068$ & $0.935$ & $0.988$ &&& $0.069$ & $0.072$ & $0.872$ & $0.993$ \\
\tKCI &&&& $0.041$ & $0.921$ &&&&& $0.075$ & $0.997$ \\
\tCP &&&& $0.095$ & $1.000$ &&&&& $0.124$ & $0.753$ \\
\bottomrule
\end{tabular*}
\caption{Test rejection rates over $N=1000$ simulations.}
\label{tab:synthetic_results}
\end{table}

\subsection{Synthetic examples} \label{sec:experiments:synthetic}

We use $\gaussian$ to denote both univariate and multivariate Gaussian distributions.
For multivariate distributions, the dimension $d$ is given in the context of each experiment.
We use $\zerovec_d$ ($\onevec_d$) to denote a vector of zeros (ones) of length $d$.

\subsubsection{Rotation} \label{sec:exp:rotation}

We first consider symmetries of random vectors $X_i \in \R^d$ with respect to the group $\grp=\SO{d}$ of $d$-dimensional rotations about the origin. 
We conduct the following experiments for $d=4$:
\begin{itemize}
    \item For tests for invariance, we estimate the test size and power by the average rejection rates over $N=1000$ simulations, each of $n=200$ i.i.d.\ samples.
    Samples were generated from $\gaussian(\zerovec_d,\I_d)$ and from $\gaussian(0.4\e_1,\I_d)$ for the null and alternative hypotheses, respectively.
    \item For tests for equivariance, we first generate samples $X_{1:n}$ of size $n=50$ from $\gaussian(\zerovec_d,\Sigma_d)$, where $\Sigma_d$ is sampled from Wishart$(\I_d,d)$. 
    For $i\in\{1,\dotsc,n\}$, we then generate $Y_i$ given $X_i$ from $\gaussian(X_i,\I_d)$ to simulate the size of the test (rejection rate under the null), and from $\gaussian(|X_i|,\I_d)$ to estimate the rejection rate under an alternative, where $|X_i|$ denotes the vector obtained by taking element-wise absolute values of $X_i$.
\end{itemize}
The estimated test sizes and powers are shown in \cref{tab:synthetic_results}.
For all $\SO{4}$-symmetries, our tests reject at a rate of approximately $\sig=0.05$ for data generated under $\nullHyp$ and reject at some noticeably higher rate for data generated under $\altHyp$.
We also compute power estimates $\hat{\power}_{n,m}$ for a single dataset using the procedure described in \cref{alg:power}.
The power estimates relatively align with the ``gold-standard'' estimates based on simulations.
Finally, we validate \cref{thm:mc:test:p} by checking the distribution of the estimated $p$-value $p_B$ across the simulations.
Histograms of those are shown in \cref{fig:pvalues}.
The resulting distributions were tested for uniformity using the Kolmogorov--Smirnov test; each plot of \cref{fig:pvalues} also displays the $p$-value of that test.
(The \tSMMD $p$-values were estimated by standard bootstrap resampling.)
The MMD-based tests using the conditional Monte Carlo method from \cref{alg:mc:pval} indeed appear to be sampling uniform $p$-values when $\nullHyp$ is true. 

\begin{figure}[!t]
\centering
\includegraphics[scale=\figscale,trim={0 9mm 0 0},clip]{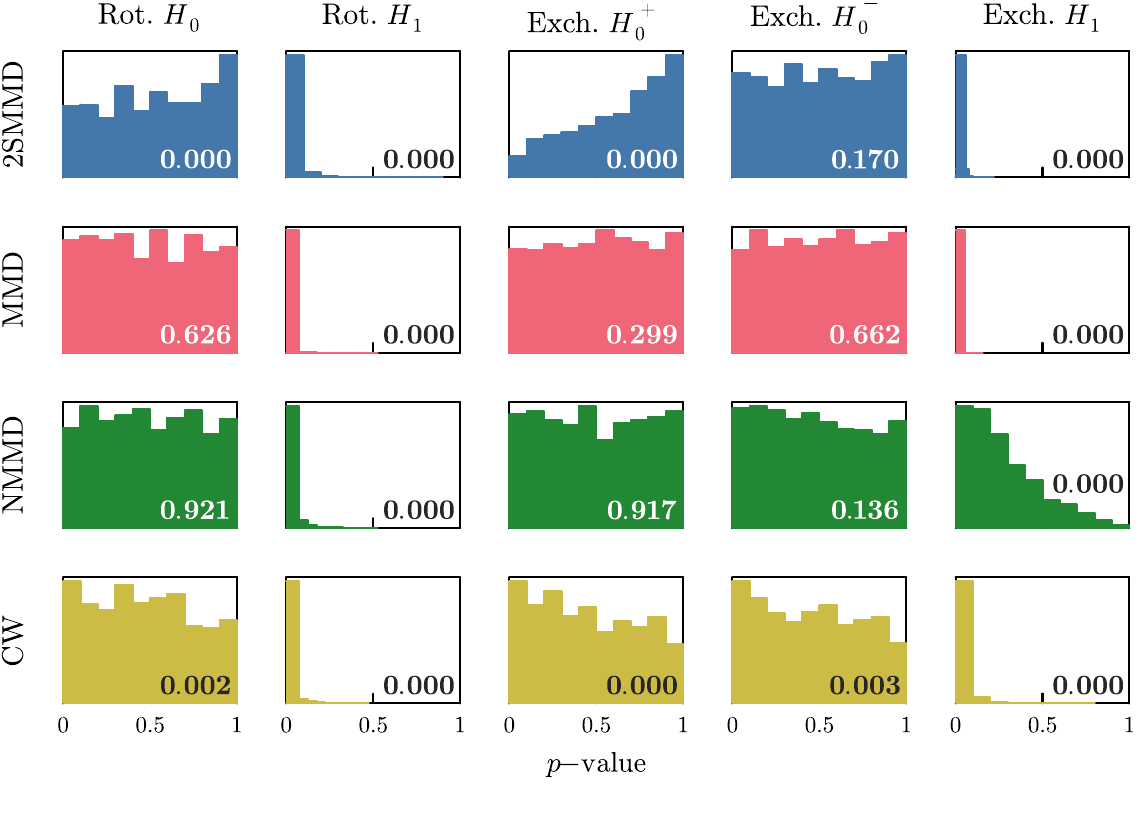}
\caption{Histograms showing the $p$-value distributions obtained over $N=1000$ simulations across different tests and data generation settings.
For \tSMMD, the $p$-value is estimated by standard bootstrapping; for the others, it is estimated by the conditional Monte Carlo method in \cref{alg:mc:pval}. 
In each plot, the $p$-value of a Kolmogorov--Smirnov test for uniformity of the distribution is shown in the bottom-right corner.}
\label{fig:pvalues}
\end{figure}

We use the $\SO{d}$-invariance experiment above to investigate how the properties of these tests change with increasing dimensions $d$, sample sizes $n$, and sampled group actions $m$.
When a variable is fixed while varying the others, we set $d=4$, $n=200$ and $m=2$.
\Cref{fig:synthetic_results} shows the estimated levels, powers and average computation (wall) times in seconds for $d\in\{5,10,15,20\}$, $n\in\{50,100,200,400\}$, and $m\in\{1,2,3,4,5\}$.
We find that among these tests for invariance, \tMMD achieves the best statistical performance as the dimension increases at the cost of increasing computation time.
\tMMD also appears to be relatively more efficient in terms of power for small sample sizes.
For smaller values of $d$, our randomized version of \tCW and the \tSMMD perform nearly as well, particularly for larger sample sizes, with less computation time. 
We also observe that sampling more than two group actions only leads to a marginal increase in power for most tests but results in an increase in computation time, most notably for \tMMD.

\begin{figure}[!t]
\centering
\includegraphics[scale=\figscale]{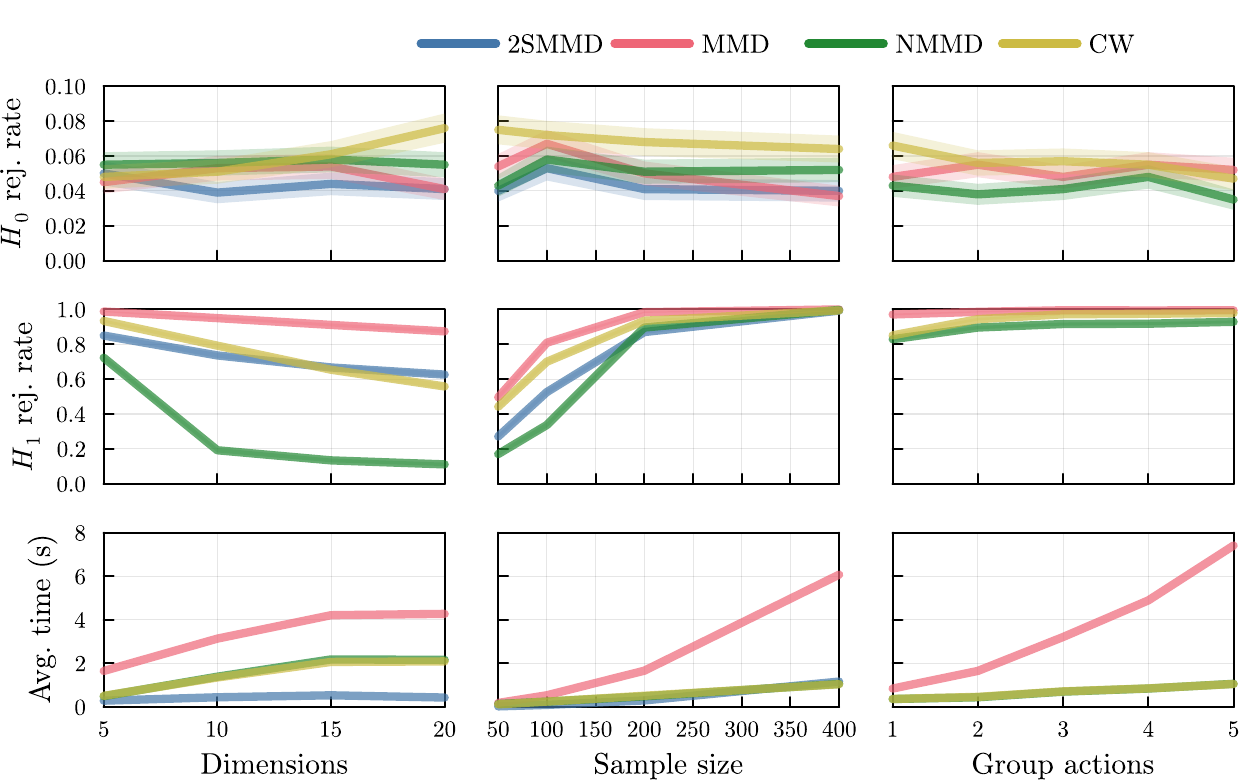}
\caption{First and second row: Test rejection rates and standard deviations over $N=1000$ simulations with one of dimensions $d$, sample size $n$, or number of group actions $m$ increasing and the others fixed ($d=4$, $n=200$, $m=2$).
Third row: Average computation time (in seconds) for a single execution of the test.}
\label{fig:synthetic_results}
\end{figure}

Separately, we also investigate a MMD test for $\SO{3}$-invariance based on representative inversions.
Given an observation $X_i\in\R^3$, we take $\repInv(X_i)$ to be the rotation matrix that maps the vector $\|X_i\|_2\e_1$ to $X_i$, obtained as in \cref{expl:SOd:invariance}.
A random inversion is then calculated as $\repInvRand_i = \repInv(X_i) G_i$, where $G_i\sim\haar_{\grp_{\e_1}}$ is a uniform random sample from the stabilizer subgroup of $\e_1$. 
We sample $X_i$ from $\gaussian(\zerovec_3,\I_3)$ to estimate the rejection rate under $\nullHyp$, and from $\gaussian(\zerovec_3,\Sigma_3)$ for $\altHyp$, where $\Sigma_3\sim\mathrm{Wishart}(\I_3,3)$.
We use a characteristic kernel on $\SO{3}$ \citep{Fukumizu:2008}, 
\begin{align*}
k(R_1, R_2) = \frac{\pi\theta(\pi-\theta)}{8\sin(\theta)} \;, \quad \text{with } \cos(\theta)=\frac{1}{2}\mathrm{Tr}(R_2^{-1} R_1)\;, \quad \text{for } 0\leq\theta\leq\pi \;.
\end{align*}
The rejection rates in the $\nullHyp$ setting are $0.014$, $0.038$, and $0.058$ for $\tSMMD$, $\tMMD$, and $\tNMMD$, respectively; the rejection rates in the $\altHyp$ setting are $0.929$, $0.976$, and $0.222$.
Histograms of the $p$-value distributions are shown in \cref{apx:exp:so3}.

\subsubsection{Exchangeability} \label{sec:exp:exch}

We next consider symmetries with respect to the group $\grp=\Sym{d}$ of permutations acting on $\R^d$. 
We conduct the following experiments for $d=10$:
\begin{itemize}
    \item For tests for invariance, in each simulation iteration, we generate i.i.d.\ samples of size $n=200$ from three distributions: $\gaussian(\zerovec_d, \Sigma_d^+)$, where $\Sigma_d^+$ is the $d\times d$ matrix with 1 on the diagonal and $\nicefrac{1}{d}$ on the off-diagonals; $\gaussian(\zerovec_d, \Sigma_d^-)$, where $\Sigma_d^-$ is similar to $\Sigma_d^+$ but with $\nicefrac{-1}{(d-1)}$ on the off-diagonals; and $\gaussian(\zerovec_d,\Sigma_d)$, where $\Sigma_d$ is randomly sampled from Wishart$(I_d,d)$.
    The first two distributions are $\Sym{d}$-invariant \citep[][pp.~7-8]{Aldous:1983}, while the third distribution is almost surely not.
    We refer to the three settings as $\nullHyp^+$, $\nullHyp^-$, and $\altHyp$, respectively.
    \item For tests for equivariance, we generate samples $X_{1:n}$ of size $n=100$ from $\gaussian(\zerovec_d,\Sigma_d)$, where $\Sigma_d$ is sampled from Wishart$(\I_d,d)$.
    For $i\in\{1,\dotsc,n\}$, we then generate $Y_i$ from $\gaussian(X_i,\I_d)$ for $\nullHyp$ and from $\gaussian(X_i^\T\e_1\onevec_d,\I_d)$ for $\altHyp$.
    For these tests, we select the kernel bandwidths via the grid search procedure described in \cref{sec:experiments}, over the grid $\{10^{-6},\dotsc,10^2\}$ for all kernels.
\end{itemize}
The results of the above tests are shown in \cref{tab:synthetic_results}.
As in the case of testing $\SO{4}$-symmetries, most of our tests for $\Sym{10}$-symmetries have rejection rates relatively close to $\sig=0.05$ when $\nullHyp$ is true and otherwise higher rates when $\altHyp$ is true.
Power estimates are again largely consistent with the simulated rejection rates. Finally, the estimated $p$-values shown in \cref{fig:pvalues} validate \cref{thm:mc:test:p} in this setting as well.
In \cref{apx:exp:exch:rand:proj}, we find that increasing the number of random projections in \tNMMD to $J=25$ boosts the test power to that comparable with \tCW.

\subsection{\textsc{Swarm} geomagnetic satellite data} \label{sec:exp:swarm}

The first data application that we consider features the geomagnetic field data collected by the European Space Agency as part of the \textsc{Swarm} constellation mission that was first launched in late 2013 \citep{Olsen:2013}.
The source dataset includes geomagnetic field measurements (in units of nT) recorded in 1-second intervals from three satellites (labelled A, B, C) that orbit the earth several times a day.
The dataset also includes the latitude, longitude, and the radius from the center of the earth at which the measurements were collected.
A magnetic dipole model, which is invariant to rotations about the dipole axis (which intersects Earth at the geomagnetic poles), can be used to approximate the geomagnetic field.
However, significant deviations from the model---and symmetry---can occur due to variations in mineral composition of Earth's crust, the effects of solar wind, and other causes. 

Following \citet{Christie:2023}, we consider only the data collected by satellite~A on February 25$^\text{th}$, 2023. 
The data consists of 86,400 measurements collected along satellite orbit trajectories around Earth.
We randomly partition the 86,400 data points into a training set and a test set of equal size.
We transform the latitude, longitude, and radius positions into Cartesian coordinates, and apply a rescaling so that the maximum norm of any observed coordinate vector is $1$.
We also standardize the geomagnetic field measurements.
We take $X_i$ to be the Cartesian coordinates and $Y_i$ to be the standardized field intensity at $X_i$.
A single sample of $n=220$ data points is visualized in \cref{fig:swarm}.
From the visualization, the magnetic field intensity appears to be nearly invariant with respect to rotations around the axis through the poles (represented as the $X_3$-axis), though some deviations from invariance are evident. 

\begin{figure}[t]
\centering
\includegraphics[scale=\figscale,trim={0 5mm 0 0},clip]{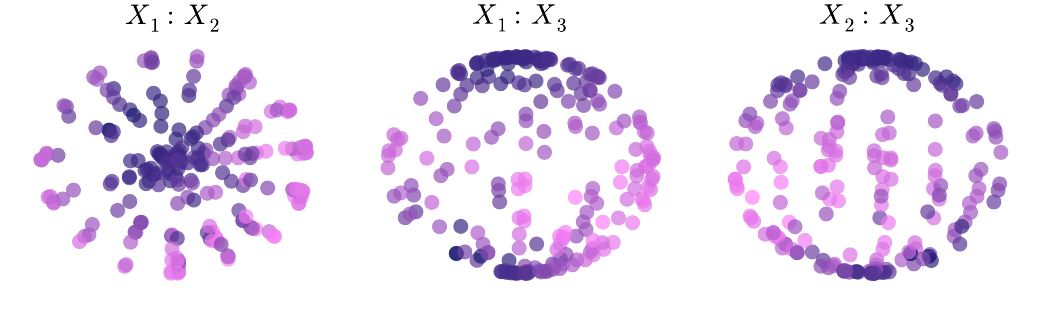}
\caption{Two-dimensional projections of $n=220$ 3D Cartesian data points sampled from the \textsc{Swarm} dataset.
The label $X_a:X_b$ indicates the projection onto the plane spanning dimensions $a$ and $b$.
The magnetic field strength is represented by the colour; the darker the colour, the greater the measured value.}
\label{fig:swarm}
\end{figure}

We test for conditional $\SO{2}$-invariance of $Y_i$ given $X_i$ with respect to each of the three axes in $\R^3$ using samples of size $n=220$ that are sampled (without replacement) from the test set.
The rejection rates are $0.464$, $0.273$, and $0.127$, where the last rate corresponds to the rotations about the geographic north pole.
We also examine conditional invariance with respect to rotations about the geomagnetic north pole,\footnote{The geomagnetic poles correspond to the axis through Earth for which a dipole approximation (which is $\SO{2}$-invariant) of the geomagnetic field is best-fitting \citep{geomag:page}.} for which \tKCI rejects at a rate of $0.223$.
\Cref{fig:pvalues:swarm} shows the distributions of the \tKCI $p$-values along with the $p$-value $0.075$ reported by \citet{Christie:2023} for their own test for invariance.
For this experiment, we were unable to tune \tCP to produce meaningful results.

We also test for the conditional invariance of $P_{Y_i|X_i}$ by testing for marginal invariance in $X_i$ and joint invariance in $(X_i,Y_i)$. 
\Cref{fig:swarm} shows that the satellite orbited the earth approximately 15 times, and the orbit trajectories are spaced approximately 24 degrees apart, intersecting at the poles.
The results are shown in \cref{apx:exp:swarm}.
We find that the tests generally do not reject marginal invariance of $X_i$ with respect to both discrete 24-degree rotations and continuous rotations about the geographic north pole.
However, \tMMD and \tCW both reject joint invariance with respect to discrete/continuous rotations in $X_i$ and non-transformations in $Y_i$ at a significantly higher rate.
Altogether, the tests indicate that $P_{Y_i|X_i}$ is not conditionally invariant with respect to rotations around any of the tested axes. 

Note that by treating the $(X_i,Y_i)$ observations as i.i.d., we are implicitly assuming that $X_{1:n}$ are i.i.d.\ and that $Y_{1:n}$ are conditionally independent given $X_{1:n}$.
The former assumption is justified by the sampling process.
The conditional independence of $Y_{1:n}$ is more difficult to justify, but deviations from that should be confined to relatively small local regions.
If one is unwilling to make the assumption, the data could be treated as a matrix $X$ of Cartesian coordinates and a vector $Y$ of corresponding geomagnetic field values.
A test for conditional invariance is possible by testing based on the maximal invariant $\maxInv(X) = X^{\top} X$.

\subsection{Large Hadron Collider dijet events} \label{sec:experiments:LHC}

The second application that we examine is based on the Large Hadron Collider (LHC) Olympics 2020 dataset \citep{Kasieczka:2021}.
The LHC dataset consists of 1.1 million simulated dijet events generated by PYTHIA \citep{Bierlich:2022aa}, a widely-used Monte Carlo generator for high-energy physics processes. A dijet event is two jets of particles that are produced by the collision of subatomic particles.  
The transverse momentum, polar angle $\phi$ and pseudorapidity $\eta$ for up to 200 jet constituents were recorded for each jet.
The Cartesian momentum of a particle in the transverse plane is represented by the pair
\begin{align*}
    p_x = p_\text{T}\cos(\phi) \;, \qquad p_y = p_\text{T}\sin(\phi) \;.
\end{align*}
The leading constituent in a jet is the particle with the largest transverse momentum in any direction.
In our analyses, we focus on the joint distribution of the two constituents with the largest transverse momenta in each event \citep[after][]{Desai:2021}.
A single observation is therefore a 4D vector $X = (p_{1_x},p_{1_y},p_{2_x},p_{2_y})$, where $p_1$ and $p_2$ correspond to the momenta of the two leading particles, respectively.
As in the previous experiment, we randomly split the dataset into a training set and a test set of equal size.
We draw samples of size $n=100$ in all of the following experiments.
Histograms of $p$-value distributions obtained from the tests can be found in \cref{apx:exp:lhc}.

\begin{figure}[!t]
\centering
\includegraphics[scale=\figscale,trim={0 6mm 0 0},clip]{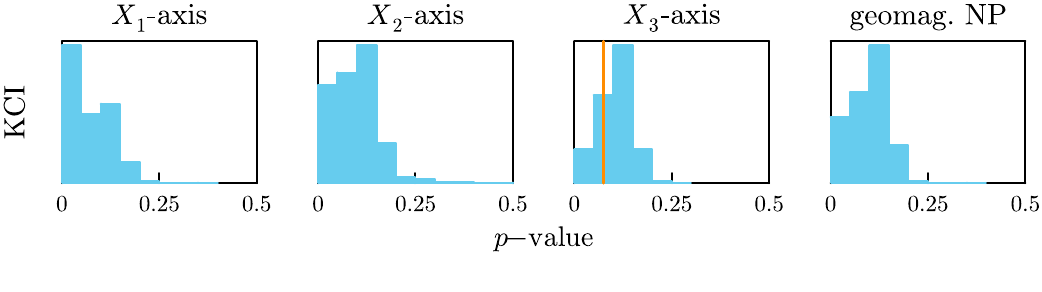}
\caption{Histograms showing the \tKCI $p$-value distributions across $N=1000$ simulations for testing conditional invariance with respect to $\SO{2}$-rotations about the three axes and the geomagnetic north pole for the \textsc{Swarm} data.
The orange line is the $p$-value ($0.075$) reported by \citet{Christie:2023} for their test for invariance.}
\label{fig:pvalues:swarm}
\end{figure}

\subsubsection{Joint invariance}

By conservation of angular momentum, the distribution of the Cartesian momenta of the two leading particles across jet events should be invariant to simultaneous rotations by the same angle, i.e., with respect to the subgroup $\grp_0=\{(g_1,g_2)\in\SO{2}\times\SO{2}:g_1=g_2\}$.
We conduct tests for invariance with respect to this subgroup, as well as with respect to the full $\grp_1=\SO{2}\times\SO{2}$ group, and to $\grp_2=\SO{4}$.
The results are shown in \cref{tab:lhc}.
We see that \tSMMD, \tMMD, and \tCW are able to identify $\grp_0$-invariance and correctly reject $\grp_1$- and $\grp_2$-invariance at a higher rate.
In \cref{apx:exp:lhc}, we find that increasing the number of random projections from 10 to 15 significantly improves the power of \tNMMD.

\subsubsection{Conditional equivariance}

By taking $X_i=(p_{1_x},p_{1_y})$ and $Y_i=(p_{2_x},p_{2_y})$, invariance of the 4D vector with respect to the subgroup $\grp_0$ can also be viewed as $Y_i$ being conditionally equivariant with respect to $\SO{2}$ given $X_i$.
We perform a test for $\SO{2}$-equivariance.
We obtain rejection rates $0.0$ for \tKCI and $0.051$ for \tCP in this setting.
We also perform a test for conditional $\SO{2}$-invariance, which \tKCI correctly rejects with rate $1.0$ and \tCP with rate $0.997$.

\subsection{Top quark tagging} \label{sec:experiments:top:quark}

We consider a second particle physics application based on the Top Quark Tagging Reference dataset \citep{Kasieczka:2019}, which also consists of jet events simulated by PYTHIA. 
The original dataset was constructed for the task of classifying jet events as having decayed from a top quark and consists of a training, validation, and test set.
We only use the test set, which contains 404,000 simulated jet events.
The four-momenta $p=(E,p_x,p_y,p_z)$ of up to 200 jet constituents are recorded for each event.
Each event is also labelled as 1 or 0, representing that the jet decayed from a top quark or did not, respectively.
As in \cref{expl:lorentz}, according to the Standard Model, when predicting whether a jet is the decay of a top quark based on the four-momenta of jet constituents, the distribution of the label should be conditionally invariant with respect to the Lorentz group $\lorentz$, which consists of spatial rotations and relativistic boosts. 
According to \cref{thm:equivariance:cond:ind}, conditional invariance is equivalent to $X \condind Y \mid \maxInv(X)$ in this scenario. 

For convenience, we take the data to be the four-momenta $X_i=(p_1,p_2)$ of the two leading constituents in each jet \citep[as in][]{Yang:2023aa} and the top quark label $Y_i\in\{0,1\}$. 
We split the data into a training and test set.
We perform a test for conditional invariance of $Y_i$ given $X_i$ with respect to the Lorentz group based on samples of size $n=200$.
In our tests, we use the 2D maximal invariant $\maxInv(X_i)=(Q(p_{1_i}),Q(p_{2_i}))$.
For the kernel on $\bfY=\{0,1\}$, we use the kernel $k_Y(x,y) = \indi(x=y)$.
\tKCI rejects conditional invariance at a rate of $0.029$, which is consistent with the theory of the Standard Model.
To verify that \tKCI is identifying symmetry in a meaningful way, we simulate new labels $Y_i'$ given $X_i$ using the model
\begin{align*}
Y_i' \mid X_i \sim \mathrm{Bernoulli}\left(0.9\indi\{E\geq 200\} + 0.1\indi\{E<200\}\right) \;.
\end{align*}
With the new labels, \tKCI rejects conditional invariance with respect to the Lorentz group at a rate of $0.781$.
Distributions of the $\tKCI$ $p$-values can be found in \cref{apx:exp:tqt}.
We were unable to tune \tCP to produce meaningful results.

\begin{table}[t]
\begin{tabular*}{\columnwidth}{@{\extracolsep\fill}rrrr@{\extracolsep\fill}}
\toprule
& $\grp_0=\{\text{paired }\SO{2}\text{-rotations}\}$ & $\grp_1=\SO{2}\times\SO{2}$ & $\grp_2=\SO{4}$ \\
\midrule
\tSMMD & $0.035$ & $0.967$ & $0.983$ \\
\tMMD & $0.038$ & $1.000$ & $1.000$ \\
\tNMMD & $0.058$ & $0.241$ & $0.214$ \\
\tCW & $0.052$ & $0.971$ & $0.999$ \\
\bottomrule
\end{tabular*}
\caption{Test rejection rates over $N=1000$ simulations for the LHC data.}
\label{tab:lhc}
\end{table}

\section*{Acknowledgments}

This research was supported in part through computational resources and services provided by Advanced Research Computing at the University of British Columbia. KC and BBR gratefully
acknowledge the support of the Natural Sciences and Engineering Research Council of Canada (NSERC): RGPIN2020-04995, RGPAS-2020-00095, DGECR-2020-00343.

\begin{appendices}

\crefalias{section}{appendix}
\crefalias{subsection}{appendix}
\crefalias{subsubsection}{appendix}

\section{Proofs of main results} \label{apx:proofs}

This section contains the proofs for \cref{thm:asymptotic:size:power,thm:mc:test:p,thm:equivariance:cond:ind}.

\subsection{Proof of \cref{thm:asymptotic:size:power}} \label{apx:proofs:size:power}

\begin{proof}
    The results follow easily from \cref{thm:invariant_characterizations} and the fact that $\metric(\empMeas{n}, \empMeasAppAvg{n}{m})$ converges almost surely (with respect to the product measure $P\otimes \haar$) to $\metric(P,\Pavg)$ by the strong law of large numbers.
    In particular, if $P \in \invProbs$, then $P = \Pavg$ and so $\metric(P,\Pavg) = 0$.
    Since $\metric$ is continuous, it follows that
    \begin{align*}
        \lim_{n\to\infty} \metric(\empMeas{n}, \empMeasAppAvg{n}{m}) = 0 \;,\quad \ \text{$P\otimes\haar$-a.s.} \;,\quad \text{for any } P \in \invProbs \;.
    \end{align*}
    Therefore, if $c_n \to c \geq 0$, then
    \begin{align*}
        \lim_{n\to\infty} \E_{P\otimes\haar}[\critFun_{n,m}(X_{1:n})] = \lim_{n\to\infty} \E_{P\otimes\haar} [\indi\{ \metric(\empMeas{n}, \empMeasAppAvg{n}{m}) > c_n \}] = 0 \;,
    \end{align*}
    from which \eqref{eq:asymp:level} follows. 

    On the other hand, if $P \in \ninvProbs$, then $P \neq \Pavg$ and therefore $\metric(P,\Pavg) > 0$.
    If $c_n \to 0$, then
    \begin{align*}
        \lim_{n\to\infty} \E_{P\otimes\haar}[\critFun_{n,m}(X_{1:n})] = \lim_{n\to\infty} \E_{P\otimes\haar} [\indi\{ \metric(\empMeas{n}, \empMeasAppAvg{n}{m}) > c_n \}] = 1 \;,
    \end{align*}
    from which \eqref{eq:asymp:power} follows. 
\end{proof}

\subsection{Proof of \cref{thm:mc:test:p}} \label{apx:proofs:p}

The proof of \cref{thm:mc:test:p} relies on the following result proven by \citet{Dufour:2006}. For simplicity, we assume that the probability of ties is zero, but that case can be handled with a randomized tie-breaking procedure described by \citet{Dufour:2006}. 

\begin{lemma}[\cite{Dufour:2006}, Proposition 2.2] \label{lem:dufour}
    Let $S_0, S_1,\dotsc,S_B$ be an exchangeable sequence of $\R$-valued random variables such that $\Pr\{S_i = S_j\} = 0$ for $i\neq j$, $i,j \in \{0,\dotsc,B\}$. Set
    \begin{align*}
        p_B = \frac{1 + \sum_{b = 1}^B \indi\{S_b \geq S_0\}}{B+1} \;.
    \end{align*}
    Then for any $\alpha \in [0,1]$,
    \begin{align*}
        \Pr\{p_B \leq \alpha\} = \frac{\lfloor \alpha (B+1) \rfloor}{B+1} \;.
    \end{align*}
\end{lemma}

\begin{proof}[Proof of \cref{thm:mc:test:p}]
    Due to the sufficiency of $\orbSel(X)$ for $\invProbs$, the samples $X^{(b)}_{1:n} = (G^{(b)}_1 X_1,\dotsc, G^{(b)}_n X_n)$ are conditionally i.i.d.\ given $\orbSel(X)_{1:n}$, with the same conditional distribution as the null conditional distribution.
    Because of their independence from $X_{1:n}$, conditioning on $(G_{j,1},\dotsc,G_{j,n})_{j=1}^m$ does not change that, and therefore $(T_{n,m}(X^{(b)}_{1:n}))_{b=0}^B$ are conditionally i.i.d.\ given $\orbSel(X)_{1:n}$ and $(G_{j,1},\dotsc,G_{j,n})_{j=1}^m$, with the same conditional distribution as the null conditional distribution.
    The sequence $(X^{(b)}_{1:n})_{b=0}^B$ is easily seen to be exchangeable (over the index $b$) conditioned on $\orbSel(X)_{1:n}$ and $(G_{j,1},\dotsc,G_{j,n})_{j=1}^m$, and therefore so is $(T_{n,m}(X^{(0)}_{1:n}),\dotsc,T_{n,m}(X_{1:n}^{(B)}))$.
    The validity of $p_B$ as a conditional (on $\orbSel(X)_{1:n}$) $p$-value, and \eqref{eq:mc:test:p:valid} in particular, follows from \cref{lem:dufour}. Since this holds for $P$-almost every realization of $\orbSel(X)_{1:n}$ under each $P \in\nullHyp$, it is also a valid $p$-value conditioned only on $(G_{j,1},\dotsc,G_{j,n})_{j=1}^m$.  

    The proof remains valid unconditionally if $(G_{j,1},\dotsc,G_{j,n})_{j=1}^m$ are sampled independently of $X_{1:n}$, so that $(X^{(b)}_{1:n}, (G_{j,1},\dotsc,G_{j,n})_{j=1}^m)_{b=0}^B$ are exchangeable and therefore so is $(T^{(0)}_{n,m}(X^{(0)}_{1:n}),\dotsc,T^{(B)}_{n,m}(X_{1:n}^{(B)}))$.  

    The proof also applies unconditionally to random $(G^{(b)}_{i,j})_{i\leq n, j\leq m}$ sampled independently of $X_{1:n}$ in a way such that they are exchangeable over the index $b=1,\dotsc,B$, in which case the sequence $(X^{(b)}_{1:n}, (G^{(b)}_{j,1},\dotsc,G^{(b)}_{j,n})_{j=1}^m)_{b=0}^B$  is exchangeable and therefore so is $(T^{(0)}_{n,m}(X^{(0)}_{1:n}),\dotsc,T^{(B)}_{n,m}(X_{1:n}^{(B)}))$. 
\end{proof}

\subsection{Proof of \cref{thm:equivariance:cond:ind}} \label{apx:proofs:equivariance}

\begin{proof}
    To simplify notation, let $Q \colon \bfX \times \bfS_{\bfY} \to [0,1]$ be a regular version (i.e., a Markov probability kernel) of the conditional probability $P_{Y|X}$, and denote the marginal distribution of $X$ by $P$, so that $P_{X,Y} = P_X \otimes P_{Y|X} = P \otimes Q$. 
    Define the random variable $\tY \define \repInvRand X$, where $\repInvRand \sim \repInvKern(X,\argdot)$. 
    The conditional distribution of $\tY$ given $(\repInvRand,X)$ is represented by the Markov probability kernel $\tQ$ so that for any integrable function $f \colon \grp \times \bfX \times \bfY \to \R$, 
    \begin{align*} %
        \int P(dx) \repInvKern(x,d\repInvRand)  \tQ(\repInvRand, x,d\ty) f(\repInvRand, x, \ty) = \int P(dx) \repInvKern(x,d\repInvRand)  Q(x,dy) f(\repInvRand, x,\repInvRand^{-1} y) \;.
    \end{align*}
    From this follows the identity $\tQ(\repInvRand, x, B) = Q(x, \repInvRand B)$. 

    Now assume that $Q$ is equivariant, so that for each $g \in \grp, x \in \bfX, B \in \bfS_{\bfY}$, $Q(gx, B) = Q(x,g^{-1}B)$.
    Then for any $\repInvRand \in \grp, x \in \bfX, g \in \grp$ and integrable $f \colon \bfY \to \R$,
    \begin{align*}
        \int  \tQ(g\repInvRand, gx,d\ty) f(\ty) & = \int Q(gx, dy) f((g\repInvRand)^{-1} y) \\
          & = \int Q(x,dy) f(\repInvRand^{-1} g^{-1} g y) \\
          & = \int Q(x,dy) f(\repInvRand^{-1}y) \\
          & = \int \tQ(\repInvRand, x, d\ty) f(\ty) \;.
    \end{align*}
    This shows that the mapping $(\repInvRand, x) \mapsto \tQ(\repInvRand, x, \argdot)$ is $\grp$-invariant.
    Therefore, by \cref{prop:orbsel:inv}, for any measurable maximal invariant $\tM \colon \grp \times \bfX \to \bfM$, there is a unique Markov probability kernel $\tR \colon \bfM \times \bfS_{\bfY} \to [0,1]$ such that
    \begin{align*} 
        \tQ(\repInvRand, x,B) = \tR(\tM(\repInvRand, x),B) \;, \quad \repInvRand \in \grp, \ x \in \bfX, \ B \in \bfS_{\bfY} \;.
    \end{align*}
    Because the action of $\grp$ on itself is transitive (i.e., there is only one orbit in $\grp$), any maximal invariant $M$ for $\grp$ acting on $\bfX$ is also a maximal invariant for $\grp$ acting on $\grp \times \bfX$, and
    \begin{align} \label{eqn:invariant:kernel}
        \tQ(\repInvRand, x,B) = \tR(\maxInv(x),B) \;, \quad \repInvRand \in \grp, \  x \in \bfX, \ B \in \bfS_{\bfY} \;.
    \end{align}
    This is enough to establish the desired conditional independence in \eqref{eqn:cond:ind:rand}: For any integrable $f \colon \grp\times \bfX \times \bfY \to \R$,
    \begin{align*}
        \int P(dx) \repInvKern(x,d\repInvRand) \tQ(\repInvRand,x,d\ty) f(\repInvRand,x, \ty) 
        & = \int P(dx)  \repInvKern(x,d\repInvRand) \tR(\maxInv(x),d\ty) f(\repInvRand,x, \ty) \;.
    \end{align*}

    Conversely, assume that $(\repInvRand,X) \condind \repInvRand^{-1} Y \mid \maxInv(X)$.
    Then \eqref{eqn:invariant:kernel} holds for $P\otimes \repInvKern$-almost all $(\repInvRand,x) \in \grp\times\bfX$.
    In particular, $\tQ$ is $\grp$-invariant for $P\otimes \repInvKern$-almost all $(\repInvRand,x)$.
    In particular, $\tQ$ is $\grp$-invariant for $P\otimes \repInvKern$-almost all $(\repInvRand,x)$.
    Recall also that the inversion kernel $\repInvKern$ is $\grp$-equivariant.
    Therefore, for any integrable $f \colon \grp \times \bfX \times \bfY \to \R$ and any $g \in \grp$,
    \begin{align*}
        \int P(dx) \repInvKern(x,d\repInvRand) Q(x,dy) f(\repInvRand,x,y) & = \int P(dx) \repInvKern(x,d\repInvRand) Q(x,dy) f(\repInvRand,x,\repInvRand (\repInvRand^{-1} y)) \\
          & = \int P(dx) \repInvKern(x,d\repInvRand) \tQ(\repInvRand, x,d\ty) f(\repInvRand,x,\repInvRand \ty) \\
          & = \int P(dx) \repInvKern(x,d\repInvRand) \tQ(g\repInvRand, gx,d\ty) f(\repInvRand,x,\repInvRand \ty) \\
          & = \int (g_* P)(dx) \repInvKern(g^{-1}x,d\repInvRand) \tQ(g\repInvRand, x,d\ty) f(\repInvRand,g^{-1}x,\repInvRand \ty) \\
          & = \int (g_* P)(dx) \repInvKern(x,d\repInvRand) \tQ(\repInvRand, x,d\ty) f(g^{-1}\repInvRand,g^{-1}x, g^{-1}\repInvRand \ty) \\
          & = \int (g_* P)(dx) \repInvKern(x,d\repInvRand) Q(x,dy) f(g^{-1}\repInvRand,g^{-1}x,g^{-1}y) \\
          & = \int P(dx) \repInvKern(gx,d\repInvRand) Q(gx,dy) f(g^{-1}\repInvRand,x,g^{-1}y) \\
          & = \int P(dx) \repInvKern(x,d\repInvRand) Q(gx,dy) f(\repInvRand,x,g^{-1}y) \;.
    \end{align*}
    This implies that
    \begin{align} \label{eqn:almost:equiv}
        Q(x,B) = Q(gx, gB) \;, \quad B \in \bfS_{\bfY}, \ g \in \grp, \ P\text{-a.e.\ } x \in \bfX \;.
    \end{align}
    The subset of $\bfX$ for which \eqref{eqn:almost:equiv} holds is a $\grp$-invariant set \citep[][Lemma 7.7]{Kallenberg:2017}, and therefore the possible exceptional null set on which $Q$ is not equivariant does not depend on $g$.
    If there is such an exceptional null set on which $Q$ is not equivariant, denoted $N^{\times}$, define $Q'$ as
    \begin{align*}
        Q'(x,B) \define
        \begin{cases}
            Q(x,B) & \text{if } x \notin N^{\times} \\
            \int_{\grp} \repInvKern(x,d\repInvRand) Q(\repInvRand^{-1}x, \repInvRand^{-1}B) & \text{if } x \in N^{\times} \;.
        \end{cases} 
    \end{align*}
    Since $\repInvKern(x,\argdot)$ and $Q(x,\argdot)$ are probability kernels, so too is $Q'$.
    It is also straightforward to show that $Q'$ is $\grp$-equivariant, so that $Q'$ is another regular version of $P_{Y|X}$ that is $\grp$-equivariant for all $x \in \bfX$, and equivalent to $Q$ up to the null set $N^{\times}$.

    If there exists a measurable representative inversion (function) $\repInv$, then the same proof holds with the inversion kernel $\repInvKern(x,\argdot)$ substituted by $\delta_{\repInv(x)}$, resulting in the simplified conditional independence statement in \eqref{eqn:cond:ind}. 

    If the action of $\grp$ on $\bfY$ is trivial, then $\tY = Y$. Moreover, $\repInvRand \condind Y \mid X$ by construction, and therefore $(\repInvRand, X) \condind Y \mid \maxInv(X)$ is implied by $X \condind Y \mid \maxInv(X)$. 
\end{proof}

\section{Alternative maximum mean discrepancy tests} \label{apx:mmd}

In this section, we describe two variations of the MMD test for invariance that have cheaper computational costs.

\subsection{Nystr\"{o}m approximation MMD test} \label{apx:mmd:nystrom}

The Nystr\"{o}m approximation \citep{Raj:2017,Chatalic:2022} can be used to obtain an approximate MMD test based on the \emph{biased} MMD test statistic, which is a V-statistic of the form
\begin{align*}
\widehat{\MMD}^{\smallsq}_\text{V}(\hat{P}_{1,n_1}, \hat{P}_{2,n_2}) &= \frac{1}{n^2}\sum_{i=1}^n\sum_{j=1}^n\left( k(X_i,X_j) + \frac{1}{m^2}\sum_{\ell=1}^m\sum_{r=1}^mk(G_\ell X_i,H_r X_j) - \frac{2}{m}\sum_{\ell=1}^m k(X_i,G_\ell X_j)\right) \\
&= \frac{1}{n^2}\left(\onevec_n^\T\bfK\onevec_n + \frac{1}{m^2}\sum_{\ell=1}^m\sum_{r=1}^m\onevec_n^\T\bfK_{\ell r}^{(2)}\onevec_n - \frac{2}{m}\sum_{\ell=1}^m\onevec_n^\T\bfK_\ell^{(1)}\onevec_n\right) \;,
\end{align*}
where the kernel matrices are defined as
\begin{align*}
\left[\bfK\right]_{ij} &= k(X_i,X_j) \;,
&
\left[\bfK_{\ell r}^{(2)}\right]_{ij} &= k(G_\ell X_i,H_r X_j) \;,
&
\left[\bfK_\ell^{(1)}\right]_{ij} &= k(X_i,G_\ell X_j) \;.
\end{align*}
Nystr\"{o}m approximates the original kernel matrices with matrix products involving $J$-dimensional random matrices.
For $J\ll n$, let $\bft$ be $J$ points sampled independently and uniformly with replacement from $\bfx \define X_{1:n}$, and similarly for $\bft^G$ from $(GX_1,\dotsc,GX_n)$.
Applying Nystr\"{o}m approximation to the MMD leads to the test statistic
\begin{align*}
\widehat{\MMD}^{\smallsq}_\text{N}(\hat{P}_{1,n_1}, \hat{P}_{2,n_2}) &= \psi_\bft^\T\bfK_{\bft,\bft}\psi_\bft + \frac{1}{m^2}\sum_{\ell=1}^m\sum_{r=1}^m\psi_{\bft^{G_\ell}}^\T\bfK_{\bft^{G_\ell},\bft^{H_r}}\psi_{\bft^{H_r}} - \frac{2}{m}\sum_{\ell=1}^m\psi_\bft^\T\bfK_{\bft,\bft^{G_\ell}}\psi_{\bft^{G_\ell}} \;,
\end{align*}
where $\bfK_{\argdot,\argdot}$ denotes the kernel matrix between two sets of points and
\begin{align*}
\psi_\argdot &= \frac{1}{n}\bfK_{\argdot,\argdot}^+\bfK_{\argdot,\bfx}\onevec_n \;,
\end{align*}
with $+$ denoting the Moore-Penrose inverse.

\subsection{MMD test with equivariant kernel} \label{apx:mmd:equiv}

If the kernel $k$ is (almost) equivariant in the sense of \eqref{eqn:equiv:kernel}, then the MMD metric in the test for invariance has the form
\begin{align*}
&\MMD(P,\Pavg) \\
&= \lip\kme{P},\kme{P}\rip_\hilbert + \int_{\bfX\times\bfX}\int_{\grp\times\grp} k(gx,hx') \haar(dg)\haar(dh)P(dx)P(dx') - 2\int_{\bfX\times\bfX} \int_{\grp} k(x, gx') \haar(dg) P(dx) P(dx') \\
&= \lip\kme{P},\kme{P}\rip_\hilbert + \int_{\bfX\times\bfX}\int_{\grp\times\grp} k(x,ghx') \haar(dg)\haar(dh)P(dx)P(dx') - 2\int_{\bfX\times\bfX} \int_{\grp} k(x, gx') \haar(dg) P(dx) P(dx') \\
&= \lip\kme{P},\kme{P}\rip_\hilbert + \int_{\bfX\times\bfX}\int_{\grp} k(x,g'x') \haar(dg')P(dx)P(dx') - 2\int_{\bfX\times\bfX} \int_{\grp} k(x, gx') \haar(dg) P(dx) P(dx') \\
&= \lip\kme{P},\kme{P}\rip_\hilbert - \int_{\bfX\times\bfX} \int_{\grp} k(x, gx') \haar(dg) P(dx) P(dx') \;,
\end{align*}
where the second equality follows from the equivariance of the kernel.
An unbiased estimator for the simplified metric is then
\begin{align*}
\widehat{\MMD}^{\smallsq}(\hat{P}_{1,n_1}, \hat{P}_{2,n_2}) = \frac{1}{n(n-1)}\sum_{i\neq j}\left(k(X_i,X_j) - \frac{1}{m}\sum_{\ell=1}^mk(X_i,G_\ell X_j)\right) \;.
\end{align*}
When the kernel satisfies \eqref{eqn:equiv:kernel}, the KME of the orbit-averaged distribution can also be interpreted as the orbit-averaged KME of the original distribution \citep[Lemma~C.2]{Elesedy:2021}, i.e., an invariant function.
We do not enforce this assumption in this work as our tests for invariance are able to operate without it.

\section{Additional experiment details} \label{apx:exp}

In this section, we provide additional details about the experiments conducted in \cref{sec:experiments}.

\subsection{Conditional permutation test with kernel conditional density estimation} \label{apx:cpt}

Let $\testStat_\text{CP}\colon\bfX^n\times\bfY^n\times\bfM^n\rightarrow\R$ be any statistic. (We use the multiple correlation coefficient of $X$ and $Y$ in our experiments in \cref{sec:experiments}.)
Let $k_Y$ and $k_\maxInv$ be kernels on $\bfY$ and $\bfM$.
Given data $X_{1:n}$ and $Y_{1:n}$, let $Z_{1:n} \define (\repInv(X)^{-1}Y)_{1:n}$ to simplify notation.
Let $Z_{\pi_0(1:n)} \define Z_{1:n}$.
On iteration~$s$, we sample $\lfloor\nicefrac{n}{2}\rfloor$ disjoint pairs of indices $(i^{(s)}_1,j^{(s)}_1),\dotsc,(i^{(s)}_{\lfloor n/2\rfloor},j^{(s)}_{\lfloor n/2\rfloor})$ from $\{1,\ldots,n\}$.
For each pair $(i^{(s)}_\ell,j^{(s)}_\ell)$, we independently perform a swap of the $i^{(s)}_\ell$-th and $j^{(s)}_\ell$-th observations with probability $p^{(s)}_\ell$ obtained from the KCDE-estimated conditional density ratio
\begin{align*}
&\frac{p^{(s)}_\ell}{1-p^{(s)}_\ell} = \frac{
\hat{f}_\text{KCDE}\left(Z_{j^{(s)}_\ell}^{(s-1)}\left|\maxInv(X_{i^{(s)}_\ell})\right.\right)
\hat{f}_\text{KCDE}\left(Z_{i^{(s)}_\ell}^{(s-1)}\left|\maxInv(X_{j^{(s)}_\ell})\right.\right)
}{
\hat{f}_\text{KCDE}\left(Z_{i^{(s)}_\ell}^{(s-1)}\left|\maxInv(X_{i^{(s)}_\ell})\right.\right)
\hat{f}_\text{KCDE}\left(Z_{j^{(s)}_\ell}^{(s-1)}\left|\maxInv(X_{j^{(s)}_\ell})\right.\right)
} \\
&= \frac{
\left\{\sum_{r=1}^n k_Y\left(Z_{j^{(s)}_\ell}^{(s-1)},Z_r\right)k_\maxInv\left(\maxInv(X_{i^{(s)}_\ell}),\maxInv(X_r)\right)\right\}
\left\{\sum_{r=1}^n k_Y\left(Z_{i^{(s)}_\ell}^{(s-1)},Z_r\right)k_\maxInv\left(\maxInv(X_{j^{(s)}_\ell}),\maxInv(X_r)\right)\right\}
}{
\left\{\sum_{r=1}^n k_Y\left(Z_{i^{(s)}_\ell}^{(s-1)},Z_r\right)k_\maxInv\left(\maxInv(X_{i^{(s)}_\ell}),\maxInv(X_r)\right)\right\}
\left\{\sum_{r=1}^n k_Y\left(Z_{j^{(s)}_\ell}^{(s-1)},Z_r\right)k_\maxInv\left(\maxInv(X_{j^{(s)}_\ell}),\maxInv(X_r)\right)\right\}
} \;.
\end{align*}
Note that this density ratio is effectively the ratio of joint densities as the denominators of the conditional density estimators cancel out.
Denote by $Z_{\pi_s(1:n)}$ the resulting permutation of $Z_{1:n}$ after all swaps in iteration $s$ have been accepted or rejected.
The CP test runs an initial $S$ iterations, after which it then runs $B$ independent sequences initialized at $Z_{\pi_S(1:n)}$, each for another $S$ iterations \citep[Algorithm~2]{Berrett:2019}.
For $b\in\{1,\dotsc,B\}$, denote the final permutation of each procedure as $Z^{(b)}_{\pi_{2S}(1:n)}$.
The $p$-value of the CP test is then computed as
\begin{align*}
p_\text{CP} = \frac{1}{1+B}\left[1+\sum_{b=1}^B\indi\left\{\testStat_\text{CP}(X_{1:n},Z_{1:n},\maxInv(X)_{1:n}) \leq \testStat_\text{CP}(X_{1:n},Z^{(b)}_{\pi_{2S}(1:n)},\maxInv(X)_{1:n})\right\} \right] \;.
\end{align*}
The test rejects $\nullHyp$ at level $\sig$ if $p_\text{CP} \leq \sig$.
See \citet{Berrett:2019} for more details about the CP test.

\subsection{Distribution of $p$-values in $\SO{3}$-invariance experiment} \label{apx:exp:so3}

\Cref{fig:pvalues:so3} shows the $p$-value distributions for the tests for $\SO{3}$-invariance conducted in \cref{sec:exp:rotation}.

\begin{figure}[!h]
    \centering
    \includegraphics[scale=\figscale,trim={0 10mm 0 0},clip]{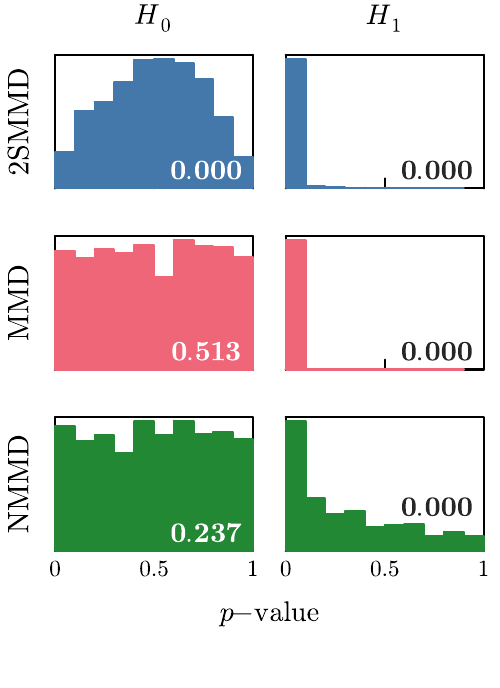}
    \caption{Histograms showing the $p$-value distributions obtained over $N=1000$ simulations in the $\SO{3}$-invariance experiment.
    The $p$-value of a Kolmogorov--Smirnov test for uniformity of the distribution is shown in the bottom-right corner of each plot.}
    \label{fig:pvalues:so3}
\end{figure}

\subsection{Number of random projections for \tNMMD and \tCW in $\Sym{10}$-invariance experiment} \label{apx:exp:exch:rand:proj}

\Cref{fig:exch:inc:proj} shows the rejection rate and average computation time for \tNMMD and \tCW as the number of random projections $J$ increases in the $\Sym{10}$-invariance experiment (\cref{sec:exp:exch}).

\begin{figure}[!ht]
\centering
\includegraphics[scale=\figscale]{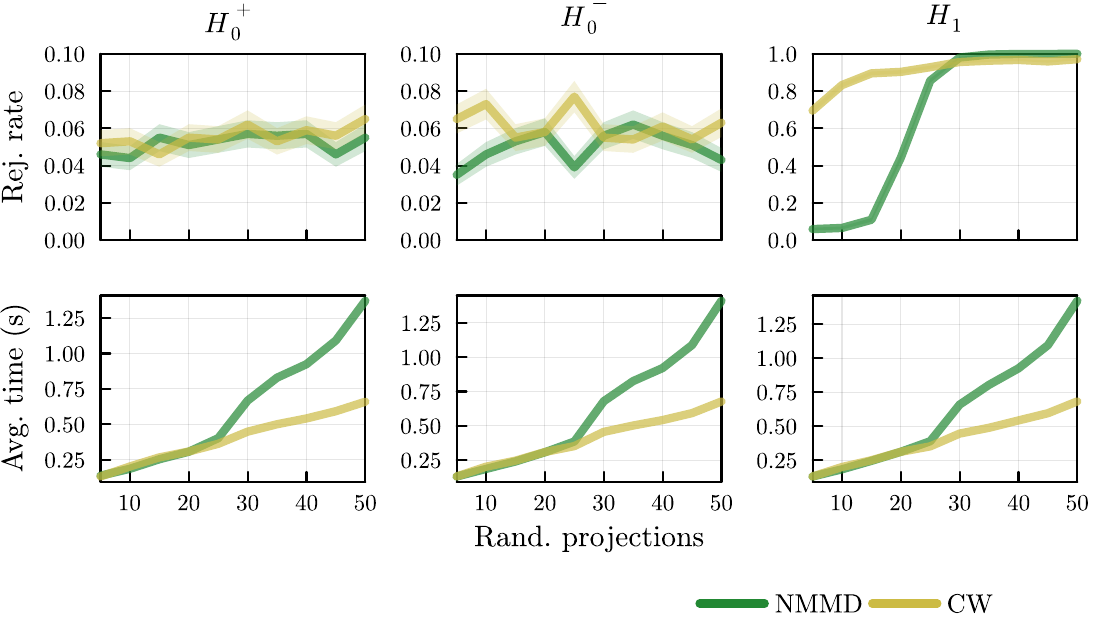}
\caption{Test for $\Sym{10}$-invariance rejection rates and standard deviations (first row) and average computation time in seconds for a single execution (second row) over $N=1000$ simulations as the number of random projections increases.}
\label{fig:exch:inc:proj}
\end{figure}

\subsection{\textsc{Swarm} experiment} \label{apx:exp:swarm}

The grids used to train the kernels in \tKCI were manually tuned through trial and error.
The final set of grids used to obtain the results in \cref{sec:exp:swarm} were
\begin{itemize}
    \item $k_X$: 31 linearly-spaced numbers between $1\text{e}^{-3}$ and $3$;
    \item $k_Y$: 57 linearly-spaced numbers between $1$ and $15$; and
    \item $k_{\maxInv(X)}$: \{$5\text{e}^{-3}$, $5.5\text{e}^{-3}$, $6\text{e}^{-3}$\}.
\end{itemize}

The test results for the marginal invariance and joint invariance experiments are given in \cref{tab:swarm}.

\begin{table}[!h]
\begin{tabular*}{\columnwidth}{@{\extracolsep\fill}rrrrr@{\extracolsep\fill}}
\toprule
& \multicolumn{2}{c}{Marginal invariance} & \multicolumn{2}{c}{Joint invariance} \\
\cmidrule(lr){2-3}
\cmidrule(l){4-5}
& Discrete & Continuous & Discrete & Continuous \\
\midrule
\tSMMD & $0.007$ & $0.004$ & $0.006$ & $0.011$ \\
\tMMD & $0.047$ & $0.047$ & $0.998$ & $0.999$ \\
\tNMMD & $0.060$ & $0.055$ & $0.098$ & $0.094$ \\
\tCW & $0.064$ & $0.087$ & $0.903$ & $0.897$ \\
\bottomrule
\end{tabular*}
\caption{Test rejection rates over $N=1000$ simulations for the \textsc{Swarm} data.}
\label{tab:swarm}
\end{table}
Histograms of the $p$-values for these tests are shown in \cref{fig:pvalues:swarm:inv}.

\begin{figure}[!ht]
\centering
\includegraphics[scale=\figscale, trim={0 10mm 0 0},clip]{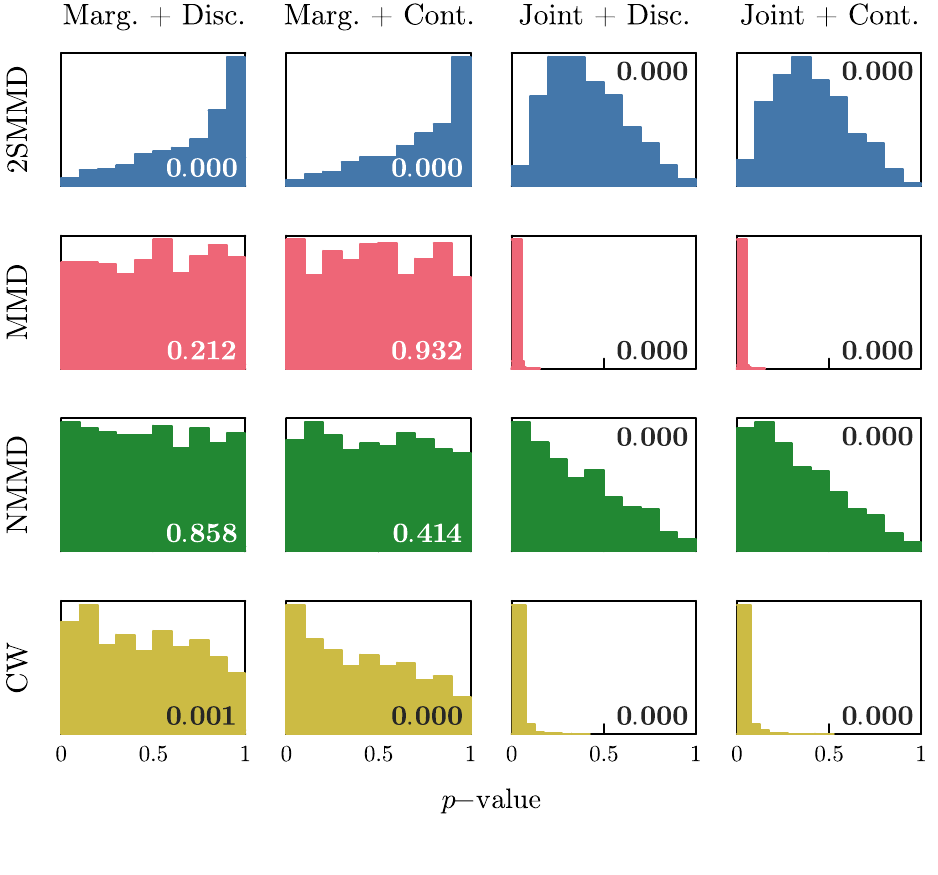}
\caption{Histograms showing the $p$-value distributions obtained over $N=1000$ simulations for tests for $X$-marginal invariance and $(X,Y)$-joint invariance with respect to discrete and continuous $X$-rotations about the geographic north pole.
    The $p$-value of a Kolmogorov--Smirnov test for uniformity of the distribution is shown in each plot.}
\label{fig:pvalues:swarm:inv}
\end{figure}

\subsection{LHC experiment} 
\label{apx:exp:lhc}

The grid $\{10^{-2},10^{-1},0,10\}$ was used to train the kernels $k_X$, $k_Y$, and $k_{\maxInv(X)}$ in \tKCI for the equivariance experiment in \cref{sec:experiments:LHC}.
The grid $\{10^{-3},10^{-2},10^{-1}\}$ was used to train the kernels $k_Y$ and $k_{\maxInv(X)}$ in $\tCP$.

\Cref{fig:lhc:pvalues} shows histograms of the $p$-values obtained from the tests for joint invariance and equivariance in \cref{sec:experiments:LHC}.

\begin{figure}[!h]
\centering
\includegraphics[scale=\figscale,trim={0 8mm 40mm 0},clip]{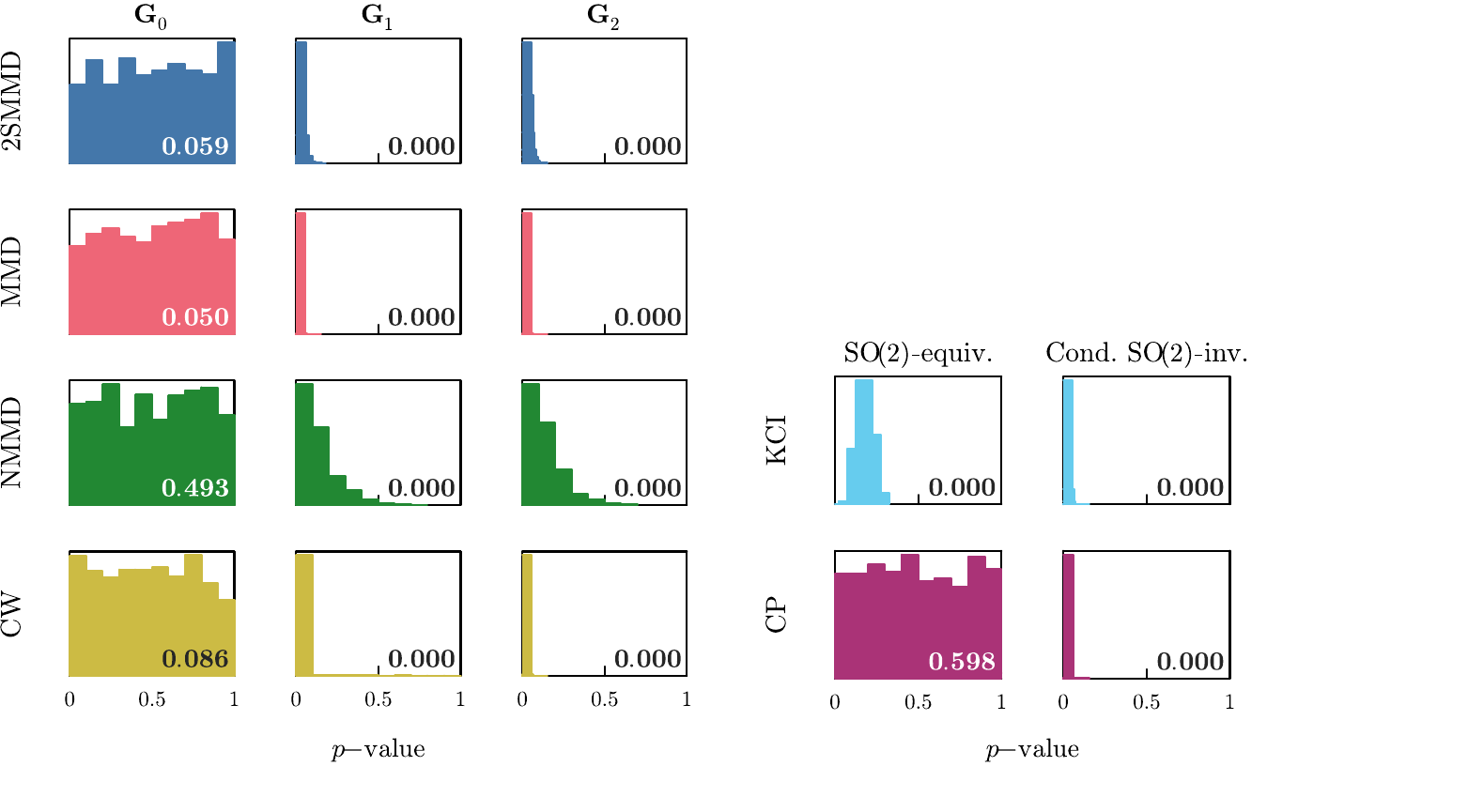}
\caption{Histograms showing the $p$-value distributions obtained over $N=1000$ simulations for tests for joint invariance and equivariance in the LHC experiments.
    The $p$-value of a Kolmogorov--Smirnov test for uniformity of the distribution is shown in the bottom-right corner of each plot.}
\label{fig:lhc:pvalues}
\end{figure}

\Cref{fig:lhc:inc:proj} shows the rejection rate and average computation time for \tNMMD and \tCW as the number of random projections $J$ increases in the LHC joint invariance experiment.

\begin{figure}[!h]
\centering
\includegraphics[scale=\figscale]{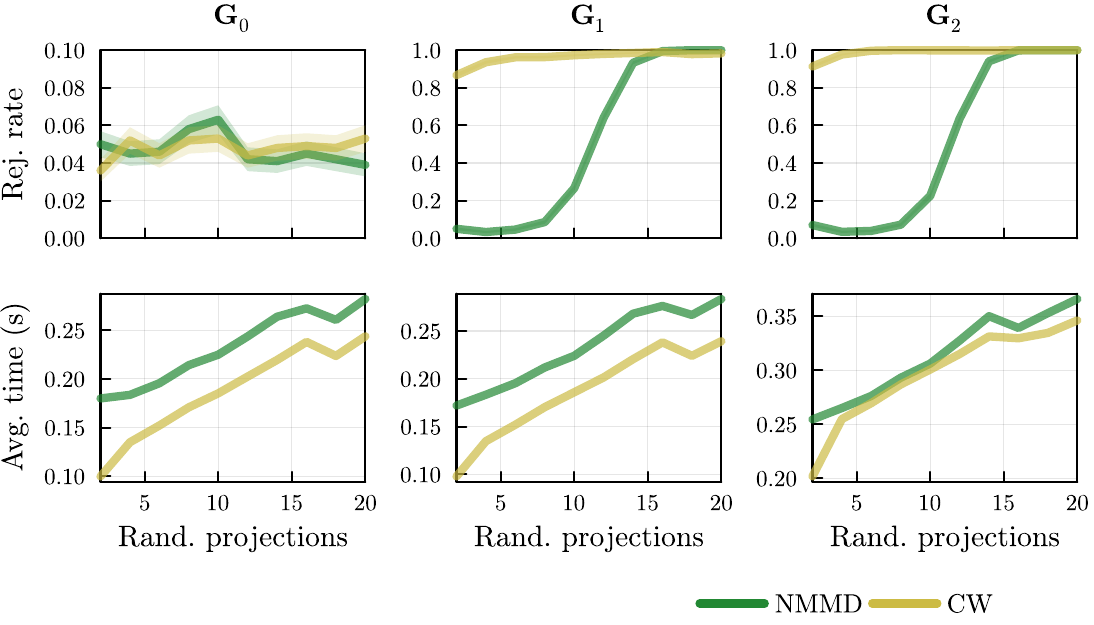}
\caption{LHC test for joint invariance rejection rates and standard deviations (first row) and average computation time in seconds for a single execution (second row) over $N=1000$ simulations as the number of random projections increases.}
\label{fig:lhc:inc:proj}
\end{figure}

\subsection{Top quark experiment} \label{apx:exp:tqt}

For \tKCI in the top quark experiment in \cref{sec:experiments:top:quark}, the grid $\{5,7.5,10,\dotsc,50\}$ was used to train the kernel $k_X$, and the grid $\{5,7.5,10,\dotsc,100\}$ was used to train the kernel $k_{\maxInv(X)}$.
The grids were manually selected based on trial and error.

\Cref{fig:pvalues:tqt} shows the $p$-value distributions obtained from \tKCI in the top quark experiment.

\begin{figure}[!h]
\centering
\includegraphics[scale=\figscale,trim={0 8mm 0 0},clip]{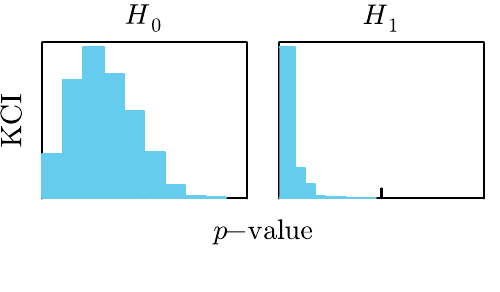}
\caption{Histograms showing the \tKCI $p$-value distributions obtained over $N=1000$ simulations in the top quark experiment.}
\label{fig:pvalues:tqt}
\end{figure}

\end{appendices}

\clearpage

\begingroup
\hypersetup{linkcolor=oceanboatblue}
\bibliography{references}
\endgroup

\end{document}

%% file: defs.tex
\def\bfH{\mathbf{H}}
\def\bfK{\mathbf{K}}
\def\bfM{\mathbf{M}}
\def\bfR{\mathbf{R}}

\def\bfX{\mathbf{X}}
\def\bfS{\mathbf{S}}
\def\bfY{\mathbf{Y}}
\def\calF{\mathcal{F}}

\def\calP{\mathcal{P}}

\def\calF{\mathcal{F}}
\def\calV{\mathcal{V}}

\def\tY{\tilde{Y}}
\def\ty{\tilde{y}}
\def\tQ{\tilde{Q}}
\def\tR{\tilde{R}}
\def\tx{\tilde{x}}

\def\bft{\mathbf{t}}
\def\bfx{\mathbf{x}}

\def\I{\mathbf{I}}

\def\R{\mathbb{R}}
\def\Pr{\mathrm{Pr}}
\def\E{\mathbb{E}}
\def\e{\mathbf{e}}
\def\metric{D}
\def\T{\top}
\def\lip{\langle}
\def\rip{\rangle}
\def\gaussian{\mathsf{N}}
\def\vMF{\mathsf{vMF}}
\newcommand{\argdot}{{\,\vcenter{\hbox{\tiny$\bullet$}}\,}}
\newcommand{\empMeas}[1]{{\hat{P}_{#1}}}

\def\zerovec{0}
\def\onevec{1}
\def\ones{\textbf{1}}
\def\indi{\mathbbm{1}} %
\def\condind{{\;\perp\!\!\!\perp\;}} %
\def\equdist{\stackrel{\text{\rm\tiny d}}{=}} %
\def\equas{\stackrel{\text{\rm\tiny a.s.}}{=}} %
\def\simiid{\sim_{\mbox{\tiny iid}}} %
\def\define{\coloneqq}

\def\bigO{\mathcal{O}}
\def\eps{\varepsilon}

\def\grp{\mathbf{G}}
\def\grpAct{\Phi}
\def\orbit{O}
\def\maxInv{M}
\def\tM{\tilde{\maxInv}}
\newcommand{\orbRep}[1]{[#1]}
\def\haar{\lambda}
\def\rhaar{\tilde{\haar}}
\def\repInv{\tau}
\def\repInvRand{\tilde{\repInv}}
\def\repInvKern{\zeta}
\def\Pavg{P^{\circ}}
\newcommand{\orbSel}{\gamma}
\def\id{\textrm{id}}
\def\proj{\pi}

\def\smallsq{\scale[0.5]{\square}}
\newcommand{\empMeasAppAvg}[2]{{\hat{P}_{#1,#2}^{\smallsq}}}
\def\dirac{\delta}

\newcommand{\SO}[1]{\mathrm{SO}(#1)}
\newcommand{\Sym}[1]{\mathbb{S}_{#1}}
\def\lorentz{\mathrm{O}(1,3)}

\def\testStat{T}
\def\nullHyp{H_0}
\def\altHyp{H_1}

\def\sig{\alpha}
\def\power{\beta}
\def\critFun{\phi}

\newcommand{\invProbs}[1][\bfX]{\mathcal{P}^{\circ}(#1)}
\newcommand{\ninvProbs}[1][\bfX]{\mathcal{P}^{\times}(#1)}
\newcommand{\probs}[1][\bfX]{\mathcal{P}(#1)}

\def\hilbert{\mathcal{H}}
\def\funeval{\varphi}
\newcommand{\kme}[1]{\mu_{#1}}
\def\MMD{\mathrm{MMD}}

\def\tSMMD{\textsc{2sMmd}\xspace}
\def\tMMD{\textsc{Mmd}\xspace}
\def\tNMMD{\textsc{Nmmd}\xspace}
\def\tCW{\textsc{Cw}\xspace}
\def\tKCI{\textsc{Kci}\xspace}
\def\tCP{\textsc{Cp}\xspace}

\newcommand{\figscale}{0.65}

\newcounter{example}
\makeatletter

\newenvironment{example}[1][\examplename]{
  \refstepcounter{example}
  \par
  \pushQED{\qed}%
  \normalfont \topsep6\p@\@plus6\p@\relax
  \begingroup
  
      \trivlist

  \item[\hskip\labelsep
        \bfseries
    #1\@addpunct{.}]\ignorespaces
}{%
  \popQED\endtrivlist\@endpefalse
  \endgroup
}
\providecommand{\examplename}{Example~\theexample}
\makeatother